\documentclass[12pt]{article}
\usepackage{amsmath}
\usepackage{amssymb}
\usepackage{bm}
\usepackage{booktabs}
\usepackage{color}
\usepackage{comment}
\usepackage{dsfont}
\usepackage{epsfig}
\usepackage{epstopdf}
\usepackage{float}
\usepackage{graphicx}
\graphicspath{ {Figures/} }
\usepackage[authoryear,round]{natbib}
\usepackage{pgfplots}
\pgfplotsset{compat=1.18}
\usepackage{setspace}
\usepackage{url}
\usepackage[dvipsnames]{xcolor}
\usepackage{microtype}

\usepackage[normalem]{ulem}
\usepackage{booktabs}

\usepackage[hidelinks]{hyperref}
\hypersetup{
  colorlinks   = true, 
  urlcolor     = blue, 
  linkcolor    = blue, 
  citecolor   = blue 
}
\usepackage{etoolbox}
\makeatletter
\g@addto@macro{\UrlBreaks}{\do\/\do-\do.\do\_}
\makeatother
\usepackage[T1]{fontenc}
\usepackage[cal=boondoxo]{mathalfa}

\newtheorem{corollary}{\bf Corollary}

\newtheorem{proposition}{\bf Proposition}

\newtheorem{lemma}{\bf Lemma}

\newenvironment{proof}[1][Proof]{\noindent\textbf{#1} }{\ \rule{0.5em}{0.5em}}
\let\ge\geqslant
\let\geq\geqslant
\let\le\leqslant
\let\leq\leqslant

\pgfplotsset{compat=1.18}
\DeclareMathOperator{\dd}{\textrm{d}\!}
\allowdisplaybreaks

\begin{document}
\title{Robust Tournaments
}
\author{Mikhail Drugov\thanks{Universitat Autònoma de Barcelona and Barcelona School of Economics (BSE), New Economic School, and CEPR, \url{mikhail.drugov@uab.cat}.} \and Dmitry Ryvkin\thanks{Economics Discipline Group, School of Economics, Finance and Marketing, RMIT University,  \url{d.ryvkin@gmail.com}.}}
\date{This version: \today}

\maketitle

\begin{abstract}
\noindent We characterize robust tournament design---the prize scheme that maximizes the lowest equilibrium effort in a rank-order tournament where the distribution of noise is unknown, except for an upper bound, $\bar{H}$, on its Shannon entropy. The robust tournament awards positive prizes to all ranks except the last, with a distinct top prize. Robust prizes spread marginal incentives as evenly as possible across performance quantiles, minimizing the Kullback--Leibler divergence from the uniform distribution to the induced distribution of performance cutoffs. Asymptotically, robust prizes follow the harmonic number sequence, while the worst-case distribution of noise converges to an exponential distribution with rate $e^{1-\bar{H}}$. The prize scheme is highly unequal, especially in small tournaments, but becomes more equitable as the number of agents grows, with the Gini coefficient approaching $1/2$. In the large tournament model, logarithmic prizes and an exponential worst-case distribution also arise under a bound on expected noise above the lower endpoint of its support, and more generally under bounds on a class of means of the quantile hazard.
\bigskip

\noindent\textbf{Keywords}: robust tournament design, max-min, prize allocation

\noindent\textbf{JEL codes}: C72, D82, D86, M52
\end{abstract}

\onehalfspacing

\section{Introduction}

Rank-based reward schemes, or \emph{tournaments}, are used extensively to incentivize effort in environments where output is noisy but relative performance is easy to observe. Examples include the allocation of bonuses and promotions in firms, sales contests, innovation races, and sports. Noise, or luck, has a significant effect on incentives in these settings, and different prize schemes can lead to very different outcomes depending on the properties of noise. An important challenge in designing rewards is, therefore, that the principal often does not know the characteristics of the noise, which can vary widely across settings. Prize schemes that are optimal when the noise distribution is known typically concentrate incentives around a small number of rank cutoffs and may be particularly vulnerable to this uncertainty. In this paper, we study \emph{robust tournament design}, where the prize structure is chosen to maximize the minimum symmetric equilibrium effort exerted by agents, under very limited information about the properties of noise.

A common assumption in the literature on prize design is that the designer has detailed knowledge of the decision environment the agents face, such as the distribution of their abilities, costs of effort, or the distribution of noise \citep[e.g.,][]{Moldovanu-Sela:2001,Fang-et-al:2020,Drugov-Ryvkin:2020_prizes}. Yet this assumption is often implausible. For example, firm executives may be unaware of the conditions workers face or the nature of their daily tasks; and university administrators may not understand the nature of shocks confronting researchers. This uncertainty is particularly relevant during structural or technological change, when past observations may no longer accurately describe the environment. For example, the introduction of AI into production processes may change how effort translates into output and the distribution of performance noise in ways that are difficult for the designer to assess, at least initially. In such cases, prize schemes that depend on precise knowledge of the environment may perform poorly. The robust design approach guarantees a lower bound on equilibrium effort across all environments in the ambiguity set.

We set up a \cite{Lazear-Rosen:1981}-type rank-order tournament model in which a budget-constrained principal chooses the prize schedule to maximize the agents' efforts. The principal does not know the noise distribution and is extremely ambiguity-averse, that is, she maximizes effort under the worst-case scenario. Without any restrictions on noise, the adversarial Nature could make the effort guarantee arbitrarily small by dilating the noise distribution. Hence, we impose an upper bound on entropy, which limits the scope for dilation while allowing considerable flexibility in the shape of the distribution.\footnote{The entropy restriction is closely connected to how tournaments generate incentives. Agents compete through output comparisons, which in a symmetric equilibrium reduce to comparisons of noise realizations. Marginal incentives arise from near ties, where a small increase in effort changes an agent’s rank, and depend on the density of noise at the relevant performance cutoffs. Shannon entropy summarizes these densities through their geometric mean across noise quantiles. An upper bound on entropy, therefore, limits how thinly Nature can spread density across the regions generating incentives. Entropy also increases under the dispersive order, which is necessary and sufficient to rank equilibrium effort across noise distributions for arbitrary tournament sizes and prize allocations \citep{Drugov-Ryvkin:2020_noise}.}

We begin our analysis with a large-tournament approximation in the spirit of \cite{Olszewski-Siegel:2016,Olszewski-Siegel:2020}. We show that the robust prize schedule has a simple logarithmic form, $w^*(z)=-\log(1-z)$, with $z$ denoting the agent's performance quantile and the average prize normalized to one. This schedule spreads budget-weighted marginal incentives uniformly across performance quantiles, preventing the adversary from targeting regions where incentives are weak. Formally, the principal's problem reduces to minimizing the Kullback--Leibler (KL) divergence from the uniform distribution to the distribution of marginal incentive weights, and the logarithmic prize schedule attains the lower bound of zero divergence.\footnote{The KL divergence, along with other more general divergences, has been used in the literature to define the ambiguity set \citep[see, e.g.,][]{Hansen-Sargent:2001,Miao-Rivera,Hwang-et-al:2026}. By contrast, in our model the KL divergence emerges as the objective function.} Uniform incentive weights are consistent with substantial prize inequality: They require larger prize increments near the top because an increment at a higher cutoff is paid to fewer agents. The corresponding worst-case distribution of noise is exponential. Note that, under $w^*$, the empirical distribution of robust prizes is exponential as well; thus, the saddle point solution is \emph{double exponential}.

The logarithmic prize schedule does not rely exclusively on the Shannon entropy bound. In Section \ref{general entropy}, we show that, in the large-tournament model, the same schedule is robustly optimal under an upper bound on expected noise above the lower endpoint of its support, as in \cite{Carroll-Meng:2016}. The worst-case distribution is again exponential. We also identify a broader class of information restrictions under which this  allocation of marginal incentives remains optimal. The Shannon and mean bounds are two economically distinct instances of this class.

We then turn to finite tournaments, for which the robust prize schedule is a unique solution to a convex programming problem. The robust tournament awards positive prizes to all ranks except the last, with a distinct top prize. This is in contrast to the optimal prize structure for a known distribution of noise, which awards some number of equal prizes at the top and zero prizes to the remaining ranks \citep{Drugov-Ryvkin:2020_prizes}. As in the large-tournament model, the principal minimizes the KL divergence from the uniform distribution, now to the distribution of performance cutoffs in quantile space. However, with finitely many agents, this distribution cannot be uniform, and the lower bound of zero divergence is not attained.

For intuition, note that tournament incentives are generated by \emph{prize differentials} between adjacent ranks: A marginal increase in effort raises an agent's earnings to the extent it helps the agent surpass a competitor, with the reward determined by the corresponding prize differential. Because the equilibrium is symmetric, output is ranked the same as the order statistics of noise. Thus, prize differentials at different ranks ``activate'' different parts of the noise distribution. When the distribution of noise is fixed, there is, generically, one optimal region where the prize differential should be created. However, concentrating incentives in this region makes the prize structure vulnerable to an adversary who reduces density there. The robust design balances incentives across regions, but adjacent order statistics overlap, so a zero prize differential does not completely shut down incentives in the corresponding region. Consequently, the robust prize schedule can contain sequences of equal prizes and is close to winner-take-all in small tournaments.

As the number of agents increases, the order statistics of noise become more concentrated, and the finite-tournament solution approaches the large-tournament benchmark. We show that, with appropriate normalization, the robust prize schedule is asymptotically approximated by a harmonic number sequence and converges to the logarithmic profile obtained in the large tournament model. For a fixed entropy bound, the corresponding adversarial distributions in the auxiliary optimization problem converge to exponential. Equilibrium implementation for each tournament size may require increasing the entropy bound with the number of agents; because the robust prize schedule is independent of this bound, this qualification does not affect the convergence of prize schedules.

These results speak to a central question in tournament design: How to allocate prizes across ranks to maximize effort. A long-standing debate concerns whether a winner-take-all (WTA) scheme, whereby the entire prize budget is awarded to the top performer, is optimal, or effort is better incentivized by awarding prizes more broadly. The answer, of course, depends on the environment. Our results suggest that substantial prize inequality associated with distributing marginal incentives uniformly across ranks can arise as a way to safeguard the principal in uncertain environments.\footnote{Quantitatively, the worst-case effort loss from choosing WTA instead of robust prizes can be substantial and increases fast with the number of agents. For example, for quadratic costs its upper bound rises from 35.5\% for $n=4$ to 98.9\% for $n=10$.} While there is no clear inequality ranking of robust prize allocations for different tournament sizes, our numerical results show that the Gini coefficient declines, i.e., prizes become more equitable, as the number of agents increases. The asymptotically optimal prize schedule has the Gini coefficient of 1/2, which corresponds to a substantial but not extreme level of inequality.

\paragraph{Related literature} Our paper contributes to a growing literature on robust mechanism design, which seeks to develop mechanisms that perform well under limited or ambiguous information. Prior work has focused on relaxing common knowledge assumptions about agents' types \citep[e.g.,][]{Bergemann-Morris:2005}, available actions \citep[e.g.,][]{Carroll:2015}, and other aspects of the environment \citep[see][for a review]{Carroll:2019}. For example, in the auction literature, \cite{Che:2022} and \cite{Suzdaltsev:2022} assume that while the distribution of values is common knowledge among the bidders, the auction designer has only partial information about it, such as a moment. We study a related setting in which the principal is uncertain about the environment itself---in particular, about the process that maps effort into observed performance. This can be viewed as asymmetric information about \emph{technology}, where the noise acts as a stochastic component of production. Our goal is to characterize tournament structures that are robust to this kind of uncertainty.\footnote{A natural interpretation of the robust tournament design approach is that of an ambiguity averse principal with the Maxmin Expected Utility preferences \citep{Gilboa-Schmeidler:1989}, whose ambiguity set consists of noise distributions satisfying the regularity condition and the entropy bound.}

Our analysis is closely related to the work on robust contracting. \emph{Linear contracts} are widely observed in practice, even though they are typically not optimal in standard moral hazard models. A substantial literature explains their optimality through various forms of robustness \citep{Hurwicz-Shapiro:1978, Holmstrom-Milgrom:1987, Edmans-Gabaix:2011, Chassang:2013, Carroll:2015, Carroll-Meng:2016, Barron-et-al:2020,Walton-Carroll:2022}. Depending on the environment, linear contracts can make incentives relatively insensitive to histories, states, realized shocks, or other features. In our large tournament setting, the optimal prize schedule is logarithmic in the rank, independently of the noise distribution. Under the worst-case distribution, which is exponential, the optimal prizes are then linear in output on the equilibrium support.\footnote{Formally, if symmetric effort is $x^*$ and stochastic output is $Y=x^*+\varepsilon$, where $\varepsilon$ is idiosyncratic noise distributed according to $F$, the logarithmic prize schedule under exponential noise becomes linear in output: $w^*(F(y-x^*)) = \lambda(y-x^*)$ in the support of $Y$.} The observed relationship between output and compensation, however,  will depend on the actual noise distribution and will generally be nonlinear.



In multi-agent settings, \cite{Dai-Toikka:2022} show that robustness to uncertainty about the game played by agents can favor contracts that align their interests, while \cite{Kambhampati:2024} shows that uncertainty about available actions can favor joint performance evaluation of independent agents. We instead take the rank-order tournament structure as given and study how uncertainty about the distribution of noise affects the allocation of prizes within it.

The impact of prize allocation on effort has been studied using various contest models, including complete-information all-pay auctions \citep[e.g.,][]{Fang-et-al:2020}, incomplete-information settings \citep[e.g.,][]{Moldovanu-Sela:2001,Zhang:2024}, tournaments with exogenous noise \citep[e.g.,][]{Drugov-Ryvkin:2020_prizes}, endogenous noise \citep{Kim-et-al:2025}, and settings with career concerns \citep[e.g.,][]{Correa-Yildirim:2024}. These models yield different conclusions---favoring WTA in some cases and recommending prize sharing in others---depending on features of the environment. In particular, \cite{Drugov-Ryvkin:2020_prizes} show, in a \cite{Lazear-Rosen:1981}-type tournament, that the optimal prize scheme is to award several equal prizes, with their number determined by the shape of the hazard rate of the noise distribution. In the presence of career concerns, \cite{Correa-Yildirim:2024} show that several prizes might be optimal though their sizes are not determined uniquely. Strictly decreasing prizes might be optimal if the agents are risk averse or have non-separable preferences over prizes and efforts \citep{Drugov-Ryvkin:2026}. Our mechanism for prize sharing is different: it arises because concentrating incentives around a small number of rank cutoffs makes the tournament vulnerable to misspecification of the noise distribution.

Our large-tournament approximation follows the approach of \cite{Olszewski-Siegel:2016,Olszewski-Siegel:2020}, in which an agent's performance determines her rank in a large population. We use this approach to isolate the mechanism of robust prize design and then show that the optimal prize schedules in the finite strategic model converge, after normalization, to the large-tournament solution.

\bigskip

The rest of the paper is structured as follows. In Section \ref{sec:continuum}, we consider robust prize design in large tournaments, where a closed-form characterization reveals the key economic mechanism. Section \ref{sec:finite} presents the full finite model, with a unique solution to the robust design problem. It also derives asymptotic prize schedules for large $n$ and shows their convergence to the large-tournament solution. Section \ref{general entropy} establishes the robustness of logarithmic prizes under more general restrictions on the distribution of noise. Section \ref{sec:conclusions} concludes. All missing proofs and additional technical results are collected in the appendix.

\section{Robust prize design in large tournaments}
\label{sec:continuum}

We begin with a large-tournament approximation similar to the approach of \cite{Olszewski-Siegel:2016,Olszewski-Siegel:2020}. This formulation is tractable and isolates the main economic mechanism of robust prize design. 

\paragraph{The environment} There is a large number of \emph{ex ante} identical agents. Each agent chooses effort $x\ge 0$ at cost $c(x)$ and produces stochastic output $Y=x+\varepsilon$, where the idiosyncratic shock $\varepsilon$ has cdf $F$ and density $f$ continuous in its support. We assume that $c$ is strictly increasing, strictly convex, and twice continuously differentiable, with $c(0)=c'(0)=0$ and $c''(x)\ge c_0>0$ for all $x\ge 0$.

Suppose that all agents except one choose effort $x^*$. When the number of competitors is large, an individual agent's effect on the empirical distribution of performance is negligible, and that empirical distribution is approximately $F(y-x^*)$. Thus, if an agent deviates to effort $x$ and receives shock $t$, her output $x+t$ corresponds to performance quantile $z=F(x-x^*+t)$. 

A \emph{prize schedule} is an absolutely continuous, increasing function $w:[0,1)\to\mathds{R}_+$, where $w(z)$ is the prize received at performance quantile $z$. We set $w(0)=0$ and impose the budget constraint $\int_0^1w(z)\dd z=1$.

Thus, the expected payoff of an agent choosing effort $x$ when all other agents choose $x^*$ is
\begin{equation}
\label{payoff_cont}
\pi(x;x^*) = \int w\left(F(x-x^*+t)\right)\dd F(t)-c(x).
\end{equation}
The symmetric first-order condition, $\pi_x(x^*;x^*)=0$, takes the form
\begin{align}
\label{foc_cont}
\int w'(F(t))f(t)\dd F(t) = c'(x^*).
\end{align}

It will be convenient to describe a prize schedule $w$ by its allocation of \emph{marginal incentive weights} across performance quantiles.
\begin{equation}
\label{q_def}
\gamma(z) = (1-z)w'(z).
\end{equation}
The budget constraint then gives, using absolute continuity of $w$ and $w(0)=0$,
\[
\int_0^1w(z)\dd z = \int_0^1\int_0^zw'(s)\dd s\dd z = \int_0^1(1-z)w'(z)\dd z = \int_0^1\gamma(z)\dd z = 1.
\]
Thus, $\gamma$ is a density over performance quantiles. We use $\Gamma$ to denote the corresponding \emph{distribution of marginal incentive weights}. For interpretation, note that $1-z$ is the fraction of agents with performance quantile above $z$. Then, for a small interval $\dd z$ around cutoff $z$, aggregate prize increment associated with cutoffs in this interval is $(1-z)w'(z)\dd z = \dd\Gamma(z)$.

Further, let $h_F(t)=\frac{f(t)}{1-F(t)}$ denote the hazard rate of noise and $Q(z)=\inf\{t:F(t)\ge z\}$ its quantile function. The first-order condition (\ref{foc_cont}) then can be written as
\begin{align}
\label{foc_cont_h}
\int_0^1 h(z)\dd \Gamma(z) = c'(x^*),
\end{align}
where $h(z)=h_F(Q(z))$ is the \emph{quantile hazard} function \citep[e.g.,][]{Nair-Sankaran:2009}. The incentives at a symmetric effort profile are thus determined by the expectation of $h$ over the distribution of marginal incentive weights.

\paragraph{Regularity constraints and equilibrium} To guarantee equilibrium implementation uniformly for our robust design problem, we impose regularity constraints over prize schedules and noise distributions. For a fixed $M>1$, let
\begin{equation}
\label{Q_cont}
\mathcal P_M = \left\{\gamma: \int_0^1\gamma(z)\dd z=1,\, \|\gamma\|_\infty\leq M, \, \log \gamma\in L^1(0,1)\right\},
\end{equation}
where $\|\gamma\|_\infty=\sup_{z\in [0,1]}\gamma(z)$.
The bound $M>1$ rules out arbitrarily concentrated marginal incentive weights $\gamma$,\footnote{It does, however, allow marginal prizes $w'$ to diverge at $z\to 1$, as long as they do not diverge faster than $(1-z)^{-1}$.} while log-integrability implies that $w'>0$ almost everywhere. As we show below, these constraints do not affect the robust solution.

Our regularity condition for noise distributions ensures that the agent’s payoff is
strictly concave in own effort, which is a standard assumption for pure-strategy equilibrium existence in \cite{Lazear-Rosen:1981}-style tournament models. Extending $f$ by zero outside its support, define the set
\begin{equation}
\label{R_cont}
\mathcal R_M = \left\{F: \frac{M}{2}\|h_F\|_\infty{\rm TV}(f)<c_0\right\}.
\end{equation}
Here, ${\rm TV}(f)$ denotes the total variation of $f$ on $\mathds{R}$.\footnote{For example, if $f$ is unimodal, we have ${\rm TV}(f)=2\|f\|_{\infty}$, i.e., total variation is twice the height of the peak of $f$.} Indeed, positive curvature in expected benefit is generated by decreasing parts of $f$, whose total downward variation is $\frac{1}{2}{\rm TV}(f)$. The hazard rate measures how strongly a marginal increase in performance changes an agent's quantile rank.\footnote{The bounded hazard requirement in $\mathcal{R}_M$ is restrictive: it excludes standard distributions such as the uniform and normal, whose hazard rates diverge toward the upper end of their support. Its advantage is that it provides a simple sufficient condition for equilibrium existence and uniqueness uniformly over all admissible prize schedules in $\mathcal P_M$. We adopt this condition here for ease of exposition. In Section \ref{sec:finite}, where we consider finite tournaments, our regularity condition allows for unbounded hazards.} We then have the following result.

\begin{lemma}
\label{lemma_equilibrium_cont}
For any $\gamma\in\mathcal P_M$ and $F\in\mathcal R_M$, the large tournament has a unique equilibrium, which is symmetric and in pure strategies, with effort $x^*$ given by (\ref{foc_cont_h}).
\end{lemma}

\paragraph{Robust prize design} We assume that the principal does not know the distribution of noise, except that it is regular and its (differential) Shannon  entropy, $H[F] = -\int \log f(t)\dd F(t)$, is bounded above by some $\bar H$.\footnote{We consider alternative restrictions on the distribution of noise in Section \ref{general entropy}.} Using the quantile hazard representation, this constraint can be written as
\begin{equation}
\label{entropy_h_cont}
\int_0^1\log h(z)\dd z \ge 1-\bar H.
\end{equation}
Thus, the ambiguity set is 
\begin{equation}
\label{F_cont}
\mathcal F_M(\bar H) = \left\{F\in\mathcal R_M: H[F]\le\bar H\right\},
\end{equation}
where we will choose parameters such that it is nonempty.

Using Lemma \ref{lemma_equilibrium_cont}, we can now write the robust design problem for the symmetric equilibrium effort in the form
\begin{align}
\label{problem_robust_cont}
\sup_{\gamma\in\mathcal{P}_M}\inf_{F\in\mathcal{F}_M(\bar{H})}\int_0^1h(z)\gamma(z)\dd z.
\end{align}
The designer chooses how to distribute marginal prize incentives across performance quantiles through $\gamma$, while adversarial Nature chooses the distribution of noise. To see the basic mechanism, temporarily ignore the regularity restriction in (\ref{problem_robust_cont}) and consider the relaxed Nature's problem for a fixed $\gamma$:\footnote{\label{fn h to F}Formally, Nature chooses $F$, and many distributions can induce the same quantile hazard. We call $h$ \emph{admissible} if $\frac{1}{[(1-z)h(z)]}$ is locally integrable in $(0,1)$. Given an admissible quantile hazard, we select the associated distribution with interval support and locally absolutely continuous quantile function $Q(z)=\int_{z_0}^z\frac{\dd s}{(1-s)h(s)}$. Here, $z_0\in(0,1)$ is any reference quantile. This distribution is determined up to a location shift, and its density satisfies $f(Q(z))=(1-z)h(z)$ almost everywhere. We can also normalize the lower bound of ${\rm supp}(F)$ to zero when it is finite.}
\begin{equation}
\label{nature_cont}
\min_h \int_0^1 h(z)\gamma(z)\dd z \qquad\text{s.t.}\qquad \int_0^1\log h(z)\dd z \ge
1-\bar H, \quad h \text{ admissible}.
\end{equation}
Jensen's inequality gives
\begin{align}
\label{Jensen_cont}
\log\left(\int_0^1 h(z) \gamma(z)\dd z\right) \ge \int_0^1\log[h(z)\gamma(z)]\dd z \ge 1-\bar H
+ \int_0^1\log \gamma(z)\dd z,
\end{align}
where the second inequality is due to the entropy constraint. The right-hand side of (\ref{Jensen_cont}) is independent of $h$; therefore, it provides the value of the relaxed problem if it is attainable. The first inequality in (\ref{Jensen_cont}) holds as equality if and only if $h(z)\gamma(z)$ is constant almost everywhere; and the second one holds as equality if the entropy constraint binds. Therefore, we can write Nature's relaxed optimum as
\begin{equation}
\label{hstar_q_cont}
h^*(z;\gamma) = \frac{\Lambda(\gamma)}{\gamma(z)}, \qquad \Lambda(\gamma) = \exp\left[1-\bar H + \int_0^1\log \gamma(s)\dd s\right].
\end{equation}

We observe that Nature assigns low hazard to precisely those quantiles at which the principal places more marginal incentive weight, so that the product $\gamma(z)h^*(z;q)$ is equalized across quantiles. Anticipating this response, the principal maximizes $\Lambda(\gamma)$ or, equivalently, she solves
\begin{equation}
\label{designer_cont}
\max_{\gamma\in\mathcal P_M}\int_0^1\log \gamma(z)\dd z.
\end{equation}
Using Jensen's inequality again, we have
\[
\int_0^1\log \gamma(z)\dd z \le \log\left(\int_0^1\gamma(z)\dd z\right) = 0,
\]
where the last step follows from the budget constraint. The inequality holds as equality if and only if $\gamma(z)=1$ almost everywhere. Thus, $\gamma^*(z)=1$---the uniform distribution of marginal incentives---provides the (relaxed) optimum for the designer.

The following proposition establishes that the same solution obtains under the regularity restriction needed for its equilibrium implementation.

\begin{proposition}
\label{prop_continuum}
Suppose $M e^{2(1-\bar H)}<c_0$. Then the unique robust marginal incentive weights are $\gamma^*(z) = 1$, $z\in[0,1]$. The corresponding prize schedule is
\begin{equation}
\label{wstar_cont}
w^*(z) = -\log(1-z), \qquad z\in[0,1).
\end{equation}
and the corresponding worst-case distribution of noise is exponential, $F^*(t)=1-e^{-\Lambda^*t}$, up to a location shift, with rate $\Lambda^*=e^{1-\bar{H}}$.
The effort guarantee satisfies $c'(x^*) = \Lambda^*$.
\end{proposition}

Proposition \ref{prop_continuum} gives a simple characterization of robust prize design: The principal allocates \emph{marginal} incentive weight uniformly across performance quantiles. Since $\gamma(z)=(1-z)w'(z)$, maintaining the same marginal incentive weight at increasingly high quantiles requires the prize schedule to become progressively steeper, producing the logarithmic schedule (\ref{wstar_cont}) that increases sharply near the top of the performance distribution. For intuition, suppose the principal places disproportionate marginal incentive weight on some performance quantiles. Nature would then respond by lowering the hazard precisely at those quantiles and raising it elsewhere. The principal eliminates this opportunity by choosing $\gamma$ uniformly and making Nature indifferent across quantiles. The worst-case hazard is then itself uniform, which corresponds to an exponential distribution of noise. The robust solution is ``double exponential'' because prizes $w^*$ are also exponentially distributed if viewed as a random variable $w^*(Z)$, with $Z\sim U_{[0,1]}$. Note that the entropy bound $\bar{H}$ does not affect the optimal prize schedule (\ref{wstar_cont}); it only scales the worst-case distribution and the effort guarantee.

Under the worst-case noise distribution, which is exponential, the robust prize schedule is linear in output: $w^*(F^*(y-x^*))=\Lambda^*(y-x^*)$ on the support of $Y$. As we mentioned in the Introduction, the literature offers various explanations for the optimality of linear contracts, each emphasizing a different form of robustness to uncertainty, e.g., about histories, states, actions, or realized shocks. In our model, however, linearity is hypothetical: the principal does not know the actual noise distribution, and the exponential distribution arises only as the worst-case scenario in her maximin problem. If the actual noise distribution has an increasing (respectively, decreasing) hazard rate, then the realized prize schedule will be convex (respectively, concave) in output.

\paragraph{Interpretation via KL divergence}
Let $I(h,\gamma)$ denote the objective in problem (\ref{problem_robust_cont}) and define distribution $G_{h,\gamma}$ on $[0,1]$ with density $g_{h,\gamma}(z) = \frac{1}{I(h,\gamma)}h(z)\gamma(z)$. Taking logs of the identity $I(h,\gamma)=\frac{h(z)\gamma(z)}{g_{h,\gamma}(z)}$ and integrating, we obtain
\[
\log I(h,\gamma) = \int_0^1\log h(z)\dd z - \int_0^1\log g_{h,\gamma}(z)\dd z + \int_0^1\log\gamma(z)\dd z,
\]
which gives
\begin{align}
\label{double_KL_cont}
\log I(h,\gamma) \ge 1-\bar{H} + {\rm D}_{\rm KL}(U_{[0,1]}\|G_{h,\gamma}) - {\rm D}_{\rm KL}(U_{[0,1]}\|\Gamma).
\end{align}
Here, the inequality follows from the entropy constraint (\ref{entropy_h_cont}), and the remaining two integrals are written as the Kullback-Leibler (KL) divergences from the uniform distribution to $G_{h,\gamma}$ and $\Gamma$, respectively.

The relaxed Nature's problem is equivalent to minimizing $\log I(h,\gamma)$ over $h$ for a fixed $\gamma$. The lower bound in (\ref{double_KL_cont}) is attained by simultaneously binding the entropy constraint and setting $g_{h,\gamma}=1$ almost everywhere. Thus, Nature \emph{minimizes} the first KL divergence by equalizing hazard-weighted incentives across quantiles, which gives $h^*(z)\gamma(z)=I(h^*,\gamma)=\Lambda(\gamma)$. Anticipating this response, the principal's maximizes $\log \Lambda(\gamma)$ over $\gamma$; therefore, the principal \emph{minimizes} the second KL divergence by choosing $\gamma=1$ almost everywhere. Hence, robustness calls for spreading marginal incentive weights as uniformly as possible across performance quantiles. Both KL divergences attain the lower bound of zero.

\section{Robust prize design in finite tournaments}
\label{sec:finite}

In this section, we consider a tournament with finitely many agents where---unlike in the large tournament setting of Section \ref{sec:continuum}---rank remains random even conditional on own performance because it depends on the realized performances of the other agents. Section \ref{sec:model} develops the full model for a finite number of agents. Section \ref{sec:robust_design} sets up the robust design problem and characterizes Nature's solution. Section \ref{sec:general charaterization} gives a general characterization of the robust prize schedule and discusses equilibrium implementation. Section \ref{sec:small_n} focuses on small tournaments. Section \ref{sec:large_n} gives an asymptotic solution for tournaments involving large but finite numbers of agents and shows that it converges to the large tournament solution of Section \ref{sec:continuum}.

\subsection{The model and preliminaries}
\label{sec:model}

\paragraph{The environment} There are now $n\ge 2$ identical, risk neutral agents indexed by $i=1,\ldots,n$. They choose efforts and produce output with the same technology as in Section \ref{sec:continuum}. The agents are rewarded with prizes ${\bf v}=(v_1,\ldots,v_n)$ based on the ranking of their outputs; that is, the agent with the highest output receives $v_1$, with the second highest---$v_2$, etc. Feasible prize schedules lie in set $\mathcal{V}=\{{\bf v}\in\mathds{R}_+^n:v_1\ge\ldots\ge v_n=0,\sum_{r=1}^nv_r=1\}$ of decreasing, non-negative prize vectors with total budget normalized to one. We assume limited liability with zero outside option, in which case it is always optimal to set the last prize $v_n$ to zero. The environment is common knowledge among the agents.

\paragraph{Regularity constraints and equilibrium} For the equilibrium characterization below, we restrict noise distributions to set
\begin{equation}
\label{regularity_class}
\mathcal R_n=\left\{F: (n-1)\|f\|_\infty{\rm TV}(f)<c_0\right\}.
\end{equation}
Similar to (\ref{R_cont}), this class restricts the curvature of the benefit of effort in order to make a deviating agent's payoff strictly concave.\footnote{In both cases, regularity controls the density of payoff-relevant performance cutoffs that determine marginal incentives. In the large-tournament approximation, the mapping from performance to rank is deterministic, and the density of cutoffs is governed by the hazard rate. In the finite tournament, rank is determined by stochastic comparisons with the $n-1$ opponents. Aggregating the curvature contributions of these stochastic rank cutoffs gives the factor $(n-1)\|f\|_\infty{\rm TV}(f)$ in the regularity condition.}

As before, we focus on symmetric, pure strategy equilibria. Assuming all agents except one exert effort $x^*$, the payoff of a deviating agent who exerts effort $x$ can be written as
\begin{align}
\label{payoff_dev}
\pi(x,x^*) = \sum_{r=1}^n\rho^{(r)}(x,x^*)v_r - c(x),
\end{align}
where
\[
\rho^{(r)}(x,x^*) = \binom{n-1}{r-1}\int F(x-x^*+t)^{n-r}[1-F(x-x^*+t)]^{r-1}\dd F(t)
\]
is the probability for the deviating agent to be ranked $r$. Payoff (\ref{payoff_dev}) is the finite-$n$ analog of (\ref{payoff_cont}). Recall that in the large tournament setting the agent's rank is deterministic conditional on the realization of noise $t$, and equal to $r\approx n(1-F(x-x^*+t))$. Here, all ranks $r=1,\ldots,n$ have positive probability given by $\rho^{(r)}$, and the agent's expected benefit is a weighted sum of prizes with these probabilities.

It is convenient to introduce cumulative probabilities ${\rm P}^{(r)} = \sum_{k=1}^r\rho^{(k)}$ for the agent to be ranked \emph{in the top} $r$. Using summation by parts, payoff (\ref{payoff_dev}) can be written as
\[
\pi(x,x^*) = \sum_{r=1}^{n-1}{\rm P}^{(r)}(x,x^*)(v_r-v_{r+1}) - c(x).
\]
Indeed, in order to be ranked in the top $r$, one needs to surpass the agent currently ranked $r$, and the reward for this is the prize differential $v_r-v_{r+1}$.

The symmetric first-order condition, $\pi_x(x^*,x^*)=0$, has the form
\begin{align}
\label{foc_byparts}
\sum_{r=1}^{n-1}B_r(v_r-v_{r+1}) = c'(x^*),
\end{align}
where coefficients $B_r={\rm P}^{(r)}_x(x^*,x^*)$---the \emph{marginal} probabilities of being ranked in the top $r$ at the symmetric effort profile---are
\begin{align}
\label{B_r_def}
B_r = r\binom{n-1}{r}\int F(t)^{n-r-1}[1-F(t)]^{r-1}f(t)\dd F(t).
\end{align}
The first-order condition (\ref{foc_byparts}) is the finite-$n$ analog of (\ref{foc_cont}) where prize differentials are replaced by $w'$. Note that $B_r>0$ for all $r=1,\ldots,n-1$ and $B_n=0$. For any ${\bf v}\in\mathcal{V}$, the left-hand side of (\ref{foc_byparts}) is positive, and hence a unique solution $x^*>0$ exists under our assumptions. The following lemma makes this first-order condition sufficient for equilibrium.

\begin{lemma}
\label{lemma_equilibrium_regularity}
For any $F\in\mathcal R_n$ and ${\bf v}\in\mathcal V$, the tournament has a unique equilibrium. The equilibrium is symmetric, in pure strategies, with effort characterized by (\ref{foc_byparts}).
\end{lemma}

\paragraph{Quantile hazard representation} Using the quantile hazard $h(z)$ introduced in Section \ref{sec:continuum}, coefficients $B_r$ in (\ref{B_r_def}) can be rewritten as
\begin{align}
\label{B_r_h}
B_r = r\binom{n-1}{r}\int_0^1 z^{n-r-1}(1-z)^rh(z)\dd z = \frac{r}{n}\mathds{E}[h(U_{(n-r:n)})],
\end{align}
where $U_{(n-r:n)}$ denotes uniform order statistic $(n-r:n)$---the $(r+1)$-st highest among $n$ i.i.d. draws from the uniform distribution on $[0,1]$.\footnote{We use the standard enumeration of order statistics $U_{(n:n)}\ge U_{(n-1:n)}\ge\ldots\ge U_{(1:n)}$ following, e.g., \cite{David-Nagaraja:2003}.} It is distributed as ${\rm Beta}(n-r,r+1)$, with pdf 
\begin{equation}
    \label{g_r}
g_r(z) = \frac{n!}{(n-1-r)!r!}z^{n-r-1}(1-z)^r, \quad r=1,\ldots,n-1, \quad z\in[0,1].
\end{equation}

Since $c(\cdot)$ is strictly convex, (\ref{foc_byparts}) implies that the impact of the distribution of noise and prize allocation on effort $x^*$ is captured entirely by the left-hand side of (\ref{foc_byparts}), which we can rewrite using (\ref{B_r_h}) as
\begin{align}
\label{MR_through_h}
\sum_{r=1}^{n-1}B_r(v_r-v_{r+1}) = \frac{1}{n}\sum_{r=1}^{n-1}p_r\mathds{E}[h(U_{(n-r:n)})] = \frac{1}{n}\int_0^1 h(z)\dd G(z;{\bf p}).
\end{align}
Here, we defined rescaled prize differences $p_r=r(v_r-v_{r+1})$; and $G(z;{\bf p})$ is the cdf corresponding to pdf $g(z;{\bf p}) = \sum_{r=1}^{n-1}p_rg_r(z)$ of the mixture of order statistics $U_{(n-r:n)}$ with weights $p_r$. 

Note that, for any ${\bf v}\in\mathcal{V}$, using summation by parts, the budget constraint can be written as $\sum_{r=1}^nv_r=\sum_{r=1}^{n-1}r(v_r-v_{r+1}) = \sum_{r=1}^{n-1}p_r = 1$. Therefore, ${\bf p}=(p_1,\ldots,p_{n-1})$ is an element of the simplex $\Delta^{n-1}$. Each $p_r$ can be interpreted as a \emph{top-$r$ budget share}---the budget share devoted to rewarding the agents for reaching rank $r$ or better. Its analog in the large tournament setting is function $\gamma(z)$. For brevity, we will refer to $p_r$ as ``budget shares.'' The original prizes can be restored from ${\bf p}$ as $v_r=\sum_{k=r}^{n-1}\frac{p_k}{k}$. 

Thus, $G(z;{\bf p})$ introduced in (\ref{MR_through_h}) is a \emph{distribution of performance cutoffs} in quantile space. It is the finite-$n$ analog of the distribution of marginal incentive weights $ \Gamma(z)$ from the large tournament model. There, rank is deterministic conditional on output, so incentive weight is allocated locally across quantile cutoffs, with density
$\gamma(z)$ that depends on the prize schedule $w$ only at $z$, cf. (\ref{q_def}). In the finite-$n$ tournament, by contrast, rank is determined by random order-statistic cutoffs that have positive density everywhere. The \emph{entire} prize schedule, summarized by the budget shares ${\bf p}$, therefore determines the mixture of these cutoff distributions that produces $G(z;{\bf p})$.

The total marginal benefit of effort in (\ref{MR_through_h}) is then obtained by interacting this distribution with the quantile hazard $h$, similar to representation (\ref{foc_cont_h}) in the large tournament setting. To understand the role of the hazard rate, recall that, by definition, 
\[ 
\mathds{P}[Q(z)<\varepsilon<Q(z)+\dd t\mid \varepsilon>Q(z)] = h(z)\dd t + o(\dd t). 
\]
That is, $h(z)$ is the density of noise immediately above $Q(z)$, conditional on the realization lying in the upper tail above $Q(z)$. Now consider $z=U_{(n-r:n)}$, the quantile of the $(r+1)$-st highest noise realization. There are $r$ realizations above $Q(z)$. The lowest of these---the $r$-th highest realization---therefore has hazard $rh(z)$ at $Q(z)$.\footnote{Indeed, hazard is the negative log-derivative of the survival function, and the survival function of the minimum of $r$ i.i.d. draws is the $r$-th power of the survival function of each individual draw.} Hence, conditional on the $(r+1)$-st highest cutoff being $z$, the marginal benefit from surpassing the agent ranked $r$ is $rh(z)(v_r-v_{r+1})=p_rh(z)$. Averaging these marginal benefits over the cutoff distributions and summing over $r$, with weight $1/n$ to account for a given agent’s probability of occupying rank $r+1$ at the symmetric effort profile, gives representation (\ref{MR_through_h}).

\subsection{Robust tournament design}
\label{sec:robust_design}

We define \emph{robust tournament} as the prize schedule ${\bf v}^*$ that maximizes the lowest symmetric equilibrium effort that can arise when the distribution of noise is unknown, except that it is regular and satisfies the entropy constraint (\ref{entropy_h_cont}). Thus, the principal knows that the distribution of noise belongs to the ambiguity set
\begin{equation}
\label{ambiguity_set}
\mathcal F_n(\bar H) = \left\{F\in\mathcal R_n: H[F]\le\bar H\right\}.
\end{equation}
We consider parameters for which this set is nonempty; Lemma \ref{lemma_implementation} below guarantees this for sufficiently large $\bar H$. Using Lemma \ref{lemma_equilibrium_regularity} and representation (\ref{MR_through_h}), we can write the robust design problem as
\begin{align}
\label{problem_maxmin}
\sup_{{\bf v}\in\mathcal{V}}\inf_{F\in\mathcal{F}_n(\bar{H})}\int_0^1h(z)g(z;{\bf p})\dd z.
\end{align}

Problem (\ref{problem_maxmin}) can be interpreted as a zero-sum game between the principal choosing ${\bf p}$ and adversarial Nature choosing $h$. For such games, under certain conditions, a version of the minimax theorem implies that the value of the problem is independent of the order of the max and min operators. The result holds in our setting under the conditions of Lemma \ref{lemma_implementation} (see Appendix \ref{app_Sion}).

Notably, problem (\ref{problem_maxmin}) has exactly the same form as (\ref{problem_robust_cont}) in the large tournament setting, and we can use a similar approach to solving it. We first consider a relaxed problem for adversarial Nature, ignoring the noise regularity constraint. Then, in the next section, we characterize the corresponding optimal prize allocation in the principal's problem and show that, for sufficiently large entropy bound $\bar{H}$, the corresponding worst-case noise distribution is regular.  

\paragraph{Nature's relaxed problem} For a given ${\bf p}\in\Delta^{n-1}$, Nature's relaxed problem in (\ref{problem_maxmin}), ignoring the noise regularity constraint, is to select a distribution of noise that solves\footnote{We can think of Nature choosing $h$ or $F$, subject to the selection of $F$ described in fn. \ref{fn h to F}.}
\begin{align}
\label{problem_min}
\inf_h \int_0^1 h(z)g(z;{\bf p})\dd z \quad \text{s.t. } \int_0^1\log h(z)\dd z\ge 1-\bar{H}, \quad h \text{ admissible}.
\end{align}
The solution is characterized in the next lemma. 

\begin{lemma}
\label{lemma h(z)}
The unique solution to Nature's relaxed problem (\ref{problem_min}) is
\begin{align}
\label{h*_p}
h^*(z;{\bf p}) = \frac{\lambda({\bf p})}{g(z;{\bf p})}, \quad \lambda({\bf p}) = \exp\left[1-\bar{H}+\int_0^1\log g(s;{\bf p})\dd s\right],
\end{align}
and the corresponding distribution $F^*(\cdot;{\bf p})$ has a bounded support.
\end{lemma}

Thus, Nature's solution has exactly the same structure as (\ref{hstar_q_cont}) in the large tournament setting. In particular, the product $h^*(z)g(z;{\bf p})$ is constant across quantiles and equals $\lambda({\bf p})$, which is the Lagrange multiplier associated with the binding entropy constraint in (\ref{problem_min}). It is also the value of Nature's minimization problem which the principal maximizes in the next step. 

An important new property in the finite $n$ case is that the adversarial noise distribution $F^*(\cdot;{\bf p})$ has bounded support. To understand the reason, note that, for any ${\bf p}\in\Delta^{n-1}$, the cutoff density $g(z;{\bf p})$ vanishes as $z\to 1$. A cutoff $z$ near the top of the performance distribution can generate incentives only if there are agents above it to be surpassed. If $r_1$ is the lowest $r$ such that $p_r>0$, this requires $r_1$ realizations to lie above $z$, which has probability of order $(1-z)^{r_1}$, and $g(z;{\bf p})$ vanishes at this rate near $z=1$. Nature can, therefore, place a diverging hazard in the upper tail and still generate only negligible incentives, which compresses the upper tail sufficiently strongly to give the worst-case distribution a finite upper bound. In the lower tail, no such compensation is needed because $1-z$ stays bounded away from zero; the integrability of the cutoff density $g(z;{\bf p})$ itself guarantees a finite lower bound.

\subsection{Robust prizes: general characterization}
\label{sec:general charaterization}

Having characterized the optimal adversarial distribution in Lemma \ref{lemma h(z)}, we now move to the maximization problem in (\ref{problem_maxmin}), whose objective is a strictly increasing transformation of $\lambda({\bf p})$ from (\ref{h*_p}). Thus, the principal's prize allocation problem is
\begin{align}
\label{problem_max}
\sup_{{\bf p}\in \Delta^{n-1}}\int_0^1\log g(z;{\bf p})\dd z,
\end{align}
which is a well-behaved convex programming problem for a strictly concave objective on a simplex. We obtain the following result.

\begin{proposition}
\label{prop_main_robust}
There exists a unique solution, ${\bf p}^*$, to problem (\ref{problem_max}), with $p_1^*>0$ and $p_{n-1}^*>0$. The corresponding prize schedule, ${\bf v}^*$, awards positive prizes to all ranks $r=1,\ldots,n-1$, with a distinct top prize.
\end{proposition}

For high-level intuition, the objective in (\ref{problem_maxmin}) can be bounded similar to (\ref{double_KL_cont}) via a combination of two KL divergences: ${\rm D}_{\rm KL}(U_{[0,1]}\|G_{h,{\bf p}})$ for Nature and $-{\rm D}_{\rm KL}(U_{[0,1]}\|G(\cdot;{\bf p}))$ for the principal. As before, Nature minimizes the divergence from the uniform distribution to the hazard-weighted marginal incentives and attains the lower bound of zero divergence via (\ref{h*_p}). The principal then minimizes the divergence from the uniform distribution to the distribution of the performance cutoffs $G(\cdot;{\bf p})$. However, unlike in large tournaments, here it is impossible to make $G$ \emph{exactly} uniform and reach the lower bound of zero divergence. The gap is due to randomness in the mapping of performance to rank, which vanishes as $n$ increases.

The properties of ${\bf p}^*$ (and ${\bf v}^*$) in Proposition \ref{prop_main_robust} follow by considering the system of Kuhn--Tucker conditions for problem (\ref{problem_max}) and verifying that $p_1>0$ (a distinct prize at the top) and $p_{n-1}>0$ (which gives $v_{n-1}>0$ and hence $v_r>0$ for all $r<n$) must hold for it to be compatible. To understand why, recall that $g(z;{\bf p})$ determines which regions of noise realizations Nature will target to minimize effort. Particularly vulnerable are the regions where $g(z;{\bf p})$ is low, as those are the regions where Nature can choose a large hazard without increasing incentives too much. For any ${\bf p}\in\Delta^{n-1}$, $g(z;{\bf p})$ is strictly positive for all $z\in(0,1)$ but can vanish at the boundaries, making the top and bottom prize differentials critical.

At the lower end, $g(0,{\bf p})=np_{n-1}$; therefore, introducing even a small $p_{n-1}>0$ has an infinite marginal benefit for the principal as compared to when $p_{n-1}=0$ and $g(z;{\bf p})$ vanishes polynomially as $z\to 0$. In contrast, the upper end is \emph{always} vulnerable because $g(1,{\bf p})=0$ for any ${\bf p}$. However, $g_r(z)\propto(1-z)^r$ near $z=1$, i.e., the top-ranked component $g_1(z)$ vanishes slower than all other components of $g(z;{\bf p})$ there. A positive $p_1$, therefore, provides the strongest protection against Nature in the upper tail. Recall from our discussion of bounded support in Section \ref{sec:robust_design} that the adversarial hazard still diverges at the upper bound, but a positive prize differential at the top helps the principal to at least slow the rate of this divergence.

\paragraph{Equilibrium implementation} Recall that we have so far considered the relaxed problem (\ref{problem_maxmin}), for which it is not guaranteed that $F^*(\cdot,{\bf p}^*)\in\mathcal{R}_n$. Proposition \ref{prop_main_robust}, however, shows that $p_1^*>0$ and $p_{n-1}^*>0$. These budget shares ensure that the worst-case noise distribution is well behaved. Indeed, $k(z) = \frac{g(z;{\bf p}^*)}{1-z}$ is a strictly positive polynomial on $[0,1]$; therefore, from (\ref{h*_p}), the worst-case density, 
\[
f^*(Q^*(z;{\bf p}^*);{\bf p}^*) = (1-z)h^*(z;{\bf p}^*) = \frac{\lambda({\bf p}^*)}{k(z)},
\]
is bounded on its compact support, and has a bounded total variation. Moreover, increasing the entropy bound from $\bar H$ to $\bar H+\Delta H$ multiplies the worst-case hazard by $e^{-\Delta H}$ and rescales (dilates) the corresponding noise distribution (see fn. \ref{fn h to F}).\footnote{Formally, optimal noise distributions corresponding to different levels of $\bar{H}$ are ranked in the \emph{dispersive order}, which unambiguously ranks the equilibrium effort across distributions \citep{Drugov-Ryvkin:2020_noise}.} Then, $\|f^*\|_\infty{\rm TV}(f^*)$ is multiplied by $e^{-2\Delta H}$. Therefore, $F^*(\cdot;{\bf p}^*)\in\mathcal{R}_n$ for $\bar{H}$ sufficiently large, and we have the following implementation result. 

\begin{lemma}
\label{lemma_implementation}
For every fixed $n\ge 2$, there exists a finite threshold $\bar{H}_n$ such that, for all $\bar H>\bar{H}_n$:

\noindent (i) the prize schedule ${\bf v}^*$ characterized in Proposition \ref{prop_main_robust} is the unique robust tournament solving (\ref{problem_maxmin});

\noindent (ii) the corresponding worst-case noise distribution $F^*=F^*(\cdot;{\bf p}^*)$ characterized in Lemma \ref{lemma h(z)} belongs to $\mathcal F_n(\bar H)$;

\noindent (iii) the corresponding effort guarantee,
\begin{equation}
\label{implemented_value}
x^*[{\bf v}^*,F^*]=(c')^{-1}\left(\frac{1}{n}e^{1-\bar H+J({\bf p}^*)}\right), \qquad J({\bf p})=\int_0^1\log g(z;{\bf p})\dd z,
\end{equation}
satisfies the saddle point property $x^*[{\bf v},F^*]\le x^*[{\bf v}^*,F^*]\le x^*[{\bf v}^*,F]$ for all ${\bf v}\in\mathcal{V}$ and $F\in\mathcal{F}_n(\bar{H})$.
\end{lemma}

For what follows, we assume that the entropy bound $\bar{H}$ is large enough for this equilibrium implementation to apply. Similar to the case of large tournaments in Section \ref{sec:continuum}, $\bar{H}$ does not affect the robust prize schedule ${\bf v}^*$ and  only scales the effort guarantee in (\ref{implemented_value}).\footnote{In Appendix \ref{app_robust_robustness}, we further show the robustness of this effort guarantee in the sense of \cite{Ball-Kattwinkel:2026}.}

We can also compare \emph{aggregate} guaranteed effort, $nx^*$, across tournament sizes. For a fixed $n$, consider a bound $\bar H$ sufficiently large for Lemma \ref{lemma_implementation} to apply to both $n$ and $n+1$. It can then be shown that for cost functions more convex than quadratic, aggregate guaranteed effort is higher in the tournament with $n+1$ agents.\footnote{This result is parallel to Proposition 6 of \cite{Drugov-Ryvkin:2020_prizes} showing that, under the same conditions, aggregate effort in an optimal tournament with a \emph{fixed} distribution of noise $F$ rises with $n$. For a proof, one can use the Minimax Theorem (which holds in our setting; see Appendix \ref{app_Sion}), and apply the minimum operator over $F$ to the pointwise inequality under optimal prizes.}

\subsection{Small tournaments}
\label{sec:small_n}

\begin{table}[t]
\centering
{\footnotesize
\begin{tabular}{c*{10}{r}}
\toprule
$n$ & $v_1^*$ & $v_2^*$ & $v_3^*$ & $v_4^*$
    & $v_5^*$ & $v_6^*$ & $v_7^*$ & $v_8^*$
    & $v_9^*$ & $v_{10}^*$ \\
\midrule
3 & 0.9055 & 0.0945 & 0 \\
\addlinespace[0.2em]
4 & 0.7979 & 0.1011 & 0.1011 & 0 \\
\noalign{\vskip -1ex}\cmidrule(lr){3-4}
5 & 0.6934 & 0.1274 & 0.1274 & 0.0518 & 0 \\
\noalign{\vskip -1ex}\cmidrule(lr){3-4}
6 & 0.6204 & 0.1304 & 0.1304 & 0.0893
  & 0.0296 & 0 \\
\noalign{\vskip -1ex}\cmidrule(lr){3-4}
7 & 0.5644 & 0.1248 & 0.1248 & 0.1248
  & 0.0340 & 0.0272 & 0 \\
\noalign{\vskip -1ex}\cmidrule(lr){3-5}
8 & 0.5107 & 0.1280 & 0.1280 & 0.1280
  & 0.0440 & 0.0440 & 0.0173 & 0 \\
\noalign{\vskip -1ex}\cmidrule(lr){3-5}\cmidrule(lr){6-7}
9 & 0.4680 & 0.1276 & 0.1276 & 0.1276
  & 0.0536 & 0.0521 & 0.0298 & 0.0138 & 0 \\
\noalign{\vskip -1ex}\cmidrule(lr){3-5}
10 & 0.4341 & 0.1244 & 0.1244 & 0.1244
   & 0.0690 & 0.0456 & 0.0456
   & 0.0211 & 0.0114 & 0 \\
\noalign{\vskip -1ex}\cmidrule(lr){3-5}\cmidrule(lr){7-8}
\bottomrule
\end{tabular}
}
\caption{Robust optimal prize schedules for $n=3,\ldots,10$, obtained by numerically solving problem (\ref{problem_max}). Groups of equal prizes are underlined.}
\label{tab:prizes}
\end{table}

In this section, we consider tournaments with a few agents. Table \ref{tab:prizes} shows the robust optimal prizes for $n=3,\ldots,10$. For $n=3$ and $4$, \cite{Krishna-Morgan:1998} established a ``winner-take-all principle,'' according to which the winner-take-all (WTA) prize schedule, with $v_1=1$, is optimal in tournaments of these sizes when the noise distribution is symmetric and unimodal.\footnote{Unimodality is not needed for that result, as shown by \cite{Drugov-Ryvkin:2020_prizes}.} In our setting, the worst-case noise distribution is not symmetric, so this principle does not apply. However, although WTA is not robustly optimal, the robust prize schedule allocates a large share of the budget to the first prize and is therefore quite close to WTA.

The following result is immediate from Proposition \ref{prop_main_robust}.

\begin{corollary}
\label{cor_n=3}
For $n=3$, the robust tournament awards distinct prizes $v_1^*>v_2^*>v_3^*=0$.
\end{corollary}
The robust optimal prizes are ${\bf v}^* \approx (0.9055,0.0945,0)$, see Table \ref{tab:prizes}. Thus, even though the WTA principle is not robust in a strict sense, it would perform relatively well in the small tournament even in settings where the principal does not know the distribution of noise. More precisely, under the worst-case distribution for the WTA as well, the loss in the equilibrium effort is no more than $\approx 7.2\%$ for quadratic costs.

Importantly, for $n=3$ the robust prize schedule has positive prize differentials \emph{everywhere}, i.e., it is fundamentally different from optimal prize schedules identified in the same setting for a fixed distribution of noise. As explained above, the latter (generically) have only \emph{one} positive prize differential corresponding to the rank $r$ that maximizes the expected hazard $\mathbb{E}[h(U_{(n-r:n)})]$.\footnote{One exception is the exponential distribution, which has a constant hazard rate. For this distribution, \emph{any} feasible prize allocation produces the same level of effort. As we saw in Section \ref{sec:continuum}, it is the worst-case distribution in the large tournament setting, see Proposition \ref{prop_continuum}. It is also the limit worst-case distribution in the finite $n$ model for $n\to\infty$, see Section \ref{sec:large_n}.} The reason, as mentioned earlier, is that it would be easy for Nature to counteract such a prize schedule, by reducing the density of noise in the corresponding region. Therefore, generally, prize differentials improve robustness, and it might seem that a robust prize schedule should have positive prize differentials for \emph{all} ranks. 

However, as seen from Table \ref{tab:prizes}, already for $n=4$ this is not the case: there is a flat region with $v_2^*=v_3^*$ (for a proof, see Appendix \ref{app:small_n}). Thus, Proposition \ref{prop_main_robust} cannot be strengthened, at least for $n$ fixed.\footnote{It can, however, be strengthened in a certain asymptotic sense, see Section \ref{sec:large_n}.} The intuition for why $v_2=v_3$ is optimal is the following. First, for any finite $n$, the densities of uniform order statistics overlap, and this overlap is stronger for $n$ small. Therefore, Nature cannot completely shut down effort by simply increasing the density near $z_r$s corresponding to zero prize differentials and keeping the density low near $z_r$s where prize differentials are positive. Besides, Nature has to respect the entropy constraint, which makes intervals of low density very costly.

Second, by Lemma \ref{lemma h(z)} the optimal adversarial distribution $F^*(t;{\bf p}^*)$ has a bounded support. Hence, it has an increasing hazard rate at least near the upper bound of the support.\footnote{For example, for $n=3$, $F^*(t;{\bf p}^*)$ is a (transformed) ${\rm Beta}(\frac{1}{2},1)$ distribution with the density
\[
f^*(t;{\bf p}^*)=
\frac{\lambda^*}{\sqrt{a^2+2b\lambda^*t}},
\qquad
t\in\left[0,\frac{a+b/2}{\lambda^*}\right],
\]
where $a=3p_2^*$, $b=6p_1^*-3p_2^*$, and
$\lambda^*=\exp\left(1-\bar H+\int_0^1\log g(z;{\bf p}^*)\,dz\right)$. It has a U-shaped hazard rate.} This calls for a significant first prize \citep[note that the winner-take-all scheme is optimal when the noise has an increasing failure rate, see][]{Drugov-Ryvkin:2020_prizes}.\footnote{Despite the significant first prize, the effort loss from choosing WTA---which is small for $n=3$ as mentioned above---becomes significant already for $n=4$ being up to 35.5\%. It increases fast with $n$ and reaches up to 98.9\% for $n=10$.} Nature then tries to minimize the density around the top order statistics and in doing so, it also makes the density low near the adjacent ranks. Recall that the principal in any case has to choose a large $p_1$ to avoid the logarithmic penalty at $z=1$. It is then optimal to forgo incentives at rank 2 by setting $p_2=0$ (unless $n=3$, in which case $p_2>0$ is needed to avoid the penalty at $z=0$). As seen from Table \ref{tab:prizes}, in all prize schedules for $n\ge 4$, there is a flat region right after the first prize. As $n$ increases, additional flat regions are formed following similar logic.

\subsection{Large finite tournaments and convergence}
\label{sec:large_n}

In this section, we characterize the robust prizes solving problem (\ref{problem_max}) for large, finite $n$ and demonstrate their convergence, in an appropriate sense, to the robust prizes $w^*$ for large tournaments in Section \ref{sec:continuum}. To make the dependence on $n$ explicit, let ${\bf p}^{*,n}$ denote the unique solution to (\ref{problem_max}); and let ${\bf v}^{*,n}$ denote the corresponding prizes. 

Consider the prize schedule $\bar{\bf v}^n$ which allocates \emph{equal} shares $\bar{p}^n_r=\frac{1}{n-1}$ of the prize budget to reward agents for reaching the top $r$, for each $r=1,\ldots,n-1$. It has the form
\begin{align}
\label{v_r theory}
\bar{v}^n_r=\frac{1}{n-1}\sum_{k=r}^{n-1}\frac{1}{k} =\frac{\chi_{n-1}-\chi_{r-1}}{n-1},
\end{align}
where $\chi_r=\sum_{k=1}^r\frac{1}{k}$ is the $r$-th harmonic number. The following result then holds.

\begin{proposition}
\label{prop_asymptotic}
As $n\to\infty$, for fixed $\bar H$,

\noindent (i) ${\bf v}^{*,n}$ converges to the harmonic prize schedule (\ref{v_r theory}), and both converge to the robust prize schedule $w^*(z)$ in (\ref{wstar_cont}) for large tournaments: for any $\alpha\in(0,1)$ and any sequence $r_n$ such that $\frac{r_n}{n-1}\to\alpha$, 
\[
(n-1)\left(v_{r_n}^{*,n}-\bar{v}_{r_n}^n\right)\longrightarrow 0 \quad \text{and} \quad (n-1)v_{r_n}^{*,n}\longrightarrow -\log\alpha;
\]

\noindent (ii) $F^*(t;{\bf p}^{*,n})$ converges weakly, up to a location shift, to the exponential distribution $F^*(t)=1-e^{-\Lambda^*t}$, $t\ge 0$.
\end{proposition}

\begin{figure}[t]
\centering
\includegraphics[width=3in]{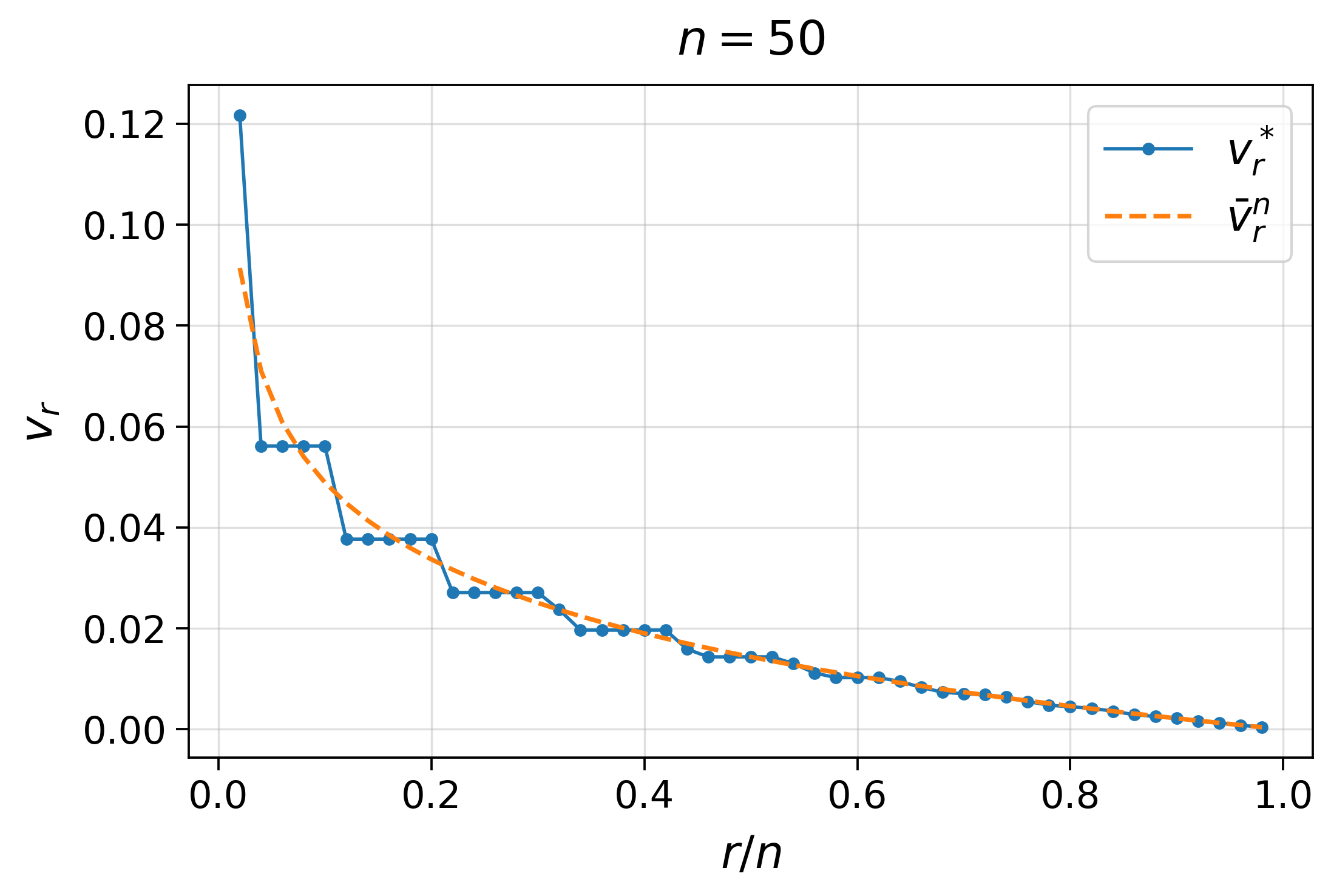}
\hfill
\includegraphics[width=3in]{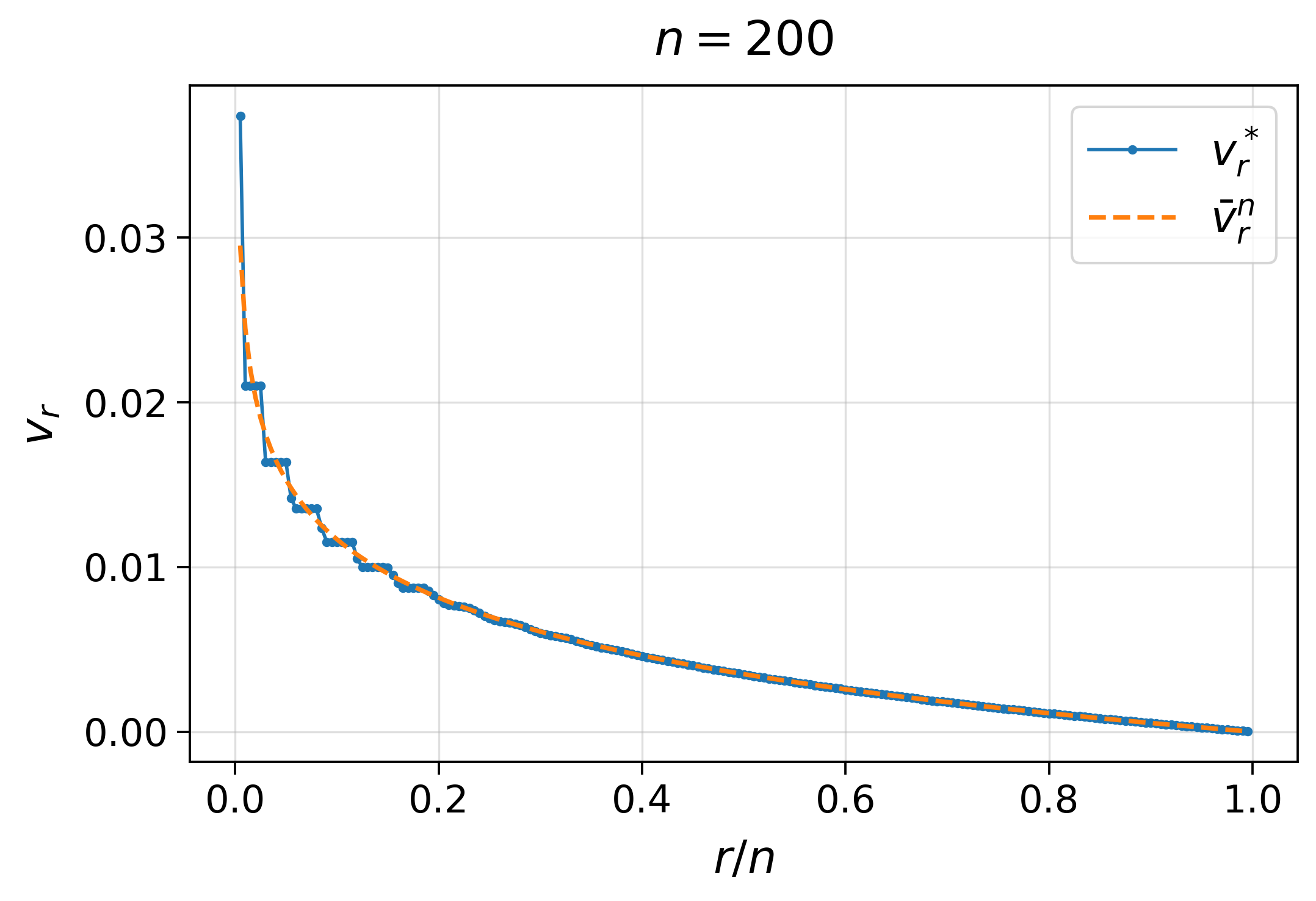}
\caption{Robust prizes as a function of reverse rank percentile $r/n$ for $n=50$ (left) and $n=200$ (right). These values are obtained by numerically solving problem (\ref{problem_max}). The dashed curve shows the asymptotic solution $\bar{\bf v}^n$, Eq. (\ref{v_r theory}).}
\label{fig:prizes_large_n}
\end{figure}

Proposition \ref{prop_asymptotic} connects robust design in the finite tournament model to the large-tournament approximation in Section \ref{sec:continuum}. This connection is not obvious a priori, since the large-tournament analysis treats rank as deterministic conditional on performance, whereas rank remains random in every finite tournament. The logarithmic prize profile emerges as the limit of the finite model, rather than as an artifact of the approximation.

The convergence of prize schedules is illustrated in Figure \ref{fig:prizes_large_n}, which shows robust optimal prizes for $n=50$ (left) and 200 (right) as functions of (reverse) ranking percentile, $\frac{r}{n}$. As seen from the figure, the asymptotic solution is a good approximation of the exact values.

Proposition \ref{prop_asymptotic} also shows that the worst-case distribution in the finite case, $F^*(t;{\bf p}^{*,n})$, converges to the worst-case distribution of the large tournament, $F^*(t)=1-e^{-\Lambda^*t}$. The exponential distribution is the only one with a constant hazard rate; thus, it is the only distribution that supports \emph{any} feasible prize allocation as effort-maximizing. The asymptotic robust prize allocation has positive prize differentials for all ranks, so the two are consistent. For a large but finite $n$, the optimal distribution of noise approaches exponential, but it has a bounded support;\footnote{It can be shown that the upper bound of the support increases as $\log(n-1)$.} and the optimal prize schedule has flat regions, cf. Figure \ref{fig:prizes_large_n}. However, the asymptotic prize schedule becomes a good approximation for $n$ large enough.

The proof of Proposition \ref{prop_asymptotic} is based on showing that the  distribution of performance cutoffs converges uniformly to the uniform distribution: $\sup_{z\in[0,1]}\left|G(z;{\bf p}^{*,n})-z\right|\longrightarrow 0$. It can then be shown that cumulative budget shares, $\sum_{k=1}^mp^{*,n}_k$, converge uniformly to $\frac{m}{n-1}$, which implies the convergence of the components of $(n-1){\bf v}^{*,n}$ to $-\log\alpha$.

An important caveat to Proposition \ref{prop_asymptotic} is that it describes the asymptotic behavior of the solution to (\ref{problem_max}) and the adversarial noise distribution for a fixed $\bar{H}$. As $n\to\infty$, the sufficient condition for equilibrium implementation established in Section \ref{sec:general charaterization} does not hold uniformly in $\bar{H}$. Instead, we can consider an increasing sequence of entropy bounds $\bar{H}_n$ large enough for the equilibrium existence condition to hold for each $n$. Since the robust prizes are independent of $\bar H$, this qualification does not affect their asymptotic structure. Therefore, $\bar{\bf v}^n$ can be thought of as an approximation of the robust prize schedule for large $n$ with a large enough entropy bound.

\section{Alternative restrictions on noise}
\label{general entropy}

In this section, we consider alternative restrictions on the distribution noise, showing that the robust tournament features we identified are not driven exclusively by the Shannon entropy bound. We first consider a bound on expected noise above the lower endpoint of its support \citep[as in][]{Carroll-Meng:2016}, and then a more general class of bounds on the quasi-arithmetic means of quantile hazard.

For ease of exposition, we focus on the large-tournament environment of Section \ref{sec:continuum} and assume that prize schedules and noise distributions are in regularity classes $\mathcal{P}_M$ and $\mathcal{R}_M$, respectively. We maintain the interval-support convention with a locally absolutely continuous quantile function, as in fn. \ref{fn h to F}.

\paragraph{Expected noise above its lower endpoint}
Suppose ${\rm supp}(F)$ is bounded below. Without loss, we can set the bound to zero and consider ambiguity set
\begin{equation}
\label{ambiguity_mean_cont}
\mathcal F_M^{\rm mean}(\mu) = \left\{F\in\mathcal R_M: \inf{\rm supp}(F)=0,\quad \mathds{E}_F[\varepsilon]\le\mu\right\}
\end{equation}
for some $\mu>0$. Thus, we consider noise distributions that are regular, have support with the lower endpoint at 0, and mean at most $\mu$.\footnote{\cite{Carroll-Meng:2016} study robust contracts for a single agent who faces additive noise and observes the shock before choosing effort. The principal knows the mean and a finite lower bound of noise, and chooses an output-based contract. An exponential distribution supports the robust solution in their model as well, although the timing of effort and contracting environment are different from ours.}

Since $Q(0+)=0$, local absolute continuity of $Q$ gives
\begin{align}
\label{mean_h_cont}
\mathds{E}_F[\varepsilon] =\int_0^1Q(z)\dd z = \int_0^1\int_0^z\frac{\dd s}{(1-s)h(s)}\dd z = \int_0^1\frac{\dd s}{h(s)}.
\end{align}
Thus, the noise mean constraint imposes a lower bound on the \emph{harmonic mean} of the quantile hazard:
\[
\left[\int_0^1\frac{\dd z}{h(z)}\right]^{-1} \ge \frac{1}{\mu}.
\]
Recall that the Shannon entropy constraint similarly imposes a lower bound on the \emph{geometric mean} of $h$, cf. (\ref{entropy_h_cont}). Both means are bounded above by the arithmetic mean, with equality when $h$ is constant.
This common property allows the principal to guarantee
marginal incentives by assigning uniform weight to performance quantiles, and suggests the following generalization.

\paragraph{Quasi-arithmetic means of quantile hazard} For a continuous and strictly monotone function $\phi:(0,\infty)\to\mathbb R$, define the \emph{quasi-arithmetic mean} (also known as the Kolmogorov-Nagumo mean),
\[
\Theta_{\phi}(h)
=\phi^{-1}\left(\int_0^1\phi(h(z))\dd z\right),
\]
where the integral is assumed to be well-defined. \emph{Power means} are a special case of this family with $\phi(u)=u$, $\log u$, and $u^{-1}$ corresponding to the arithmetic, geometric, and harmonic mean, respectively.

Consider the ambiguity set
\[
\mathcal F_{M,\phi}(\lambda)
=\{F\in\mathcal R_M:\Theta_{\phi}(h)\geq\lambda\}, \qquad \lambda>0.
\]
Similar to the special cases, this restriction limits the scope for dilation without restricting the location of noise. The following proposition shows that the logarithmic prizes and exponential worst-case distribution remain optimal.

\begin{proposition}
\label{prop_Kolmogorov}
Suppose $\phi$ is either strictly increasing and concave or strictly decreasing and convex, $M>1$, and $M\lambda^2<c_0$. Then $\gamma^*(z)=1$ solves the robust design problem with ambiguity set $\mathcal F_{M,\phi}(\lambda)$, the corresponding prize schedule is $w^*(z)=-\log(1-z)$, the exponential distribution with rate $\lambda$ is a worst-case distribution, and the effort guarantee satisfies $c'(x^*)=\lambda$.
\end{proposition}

For a proof, we can follow the argument similar to Proposition \ref{prop_continuum}. Jensen's inequality gives
\[
\int_0^1h(z)\dd z \ge \Theta_{\phi}(h) \ge \lambda.
\]
Therefore, $\gamma^*=1$ guarantees marginal incentives
of at least $\lambda$. Now consider the exponential distribution with rate $\lambda$. Its quantile hazard is $h^*=\lambda$, so
$\Theta_\phi(h^*)=\lambda$. Moreover, $\|h_{F^*}\|_\infty=\lambda$ and, extending the density by zero outside its support, ${\rm TV}(f^*)=2\lambda$. Hence,
\[
\frac{M}{2}\|h_{F^*}\|_\infty{\rm TV}(f^*) = M\lambda^2<c_0,
\]
so this distribution belongs to the ambiguity set. Moreover, for every $\gamma\in\mathcal P_M$,
\[
\int_0^1h^*(z)\gamma(z)\dd z=\lambda.
\]
These inequalities establish the saddle point. Finally,
$\gamma^*=1$ produces the logarithmic prizes, and the effort guarantee follows from Lemma \ref{lemma_equilibrium_cont}.

Applying Proposition \ref{prop_Kolmogorov} with $\phi(u)=u^{-1}$ and $\lambda=1/\mu$ shows that
logarithmic prizes are robustly optimal under the noise mean constraint (\ref{ambiguity_mean_cont}), provided
$M/\mu^2<c_0$. The exponential distribution with mean $\mu$
is a worst case, and the effort guarantee satisfies
$c'(x^*)=1/\mu$. Uniqueness follows by the same perturbation argument as in the proof of Proposition \ref{prop_continuum}.

\paragraph{Interpretation via divergence for power means}
The double-divergence representation of the effort guarantee (\ref{double_KL_cont}) extends to power means of order $\tau<1$, with the geometric mean corresponding to $\tau=0$.
As in Section \ref{sec:continuum}, let $G_{h,\gamma}$ have density $g_{h,\gamma}=h\gamma/I(h,\gamma)$, and define distribution $G_\gamma^*$ by its density
\[
g_\gamma^*(z) = \frac{\gamma(z)^{-\tau/(1-\tau)}}{\int_0^1\gamma(s)^{-\tau/(1-\tau)}\dd s}.
\]
Assuming all the integrals are finite, we then have
\[
\log I(h,\gamma) \ge \log\lambda + {\rm D}_{1-\tau}(G_\gamma^*\|G_{h,\gamma}) - {\rm D}_{1/(1-\tau)}(U_{[0,1]}\|\Gamma),
\]
where ${\rm D}_{q}(\cdot\|\cdot)$ denotes the \emph{R\'{e}nyi divergence} of order $q$. In the relaxed problem, Nature attains the lower bound by binding the constraint and setting $G_{h,\gamma}=G_\gamma^*$, provided the corresponding hazard is admissible. The principal then minimizes ${\rm D}_{1/(1-\tau)}(U_{[0,1]}\|\Gamma)$ by choosing $\gamma=1$. At this allocation, $G_\gamma^*=U$ and the worst-case hazard is constant. For $\tau=0$, both divergences become KL divergences and the representation reduces to (\ref{double_KL_cont}).\footnote{The power means are the only homogeneous quasi-arithmetic means, and this is the key property underlying the representation via R\'{e}nyi divergence. Homogeneity ensures that rescaling the noise distribution
changes the ambiguity bound proportionally without changing the form of the restriction.}

\section{Conclusions}
\label{sec:conclusions}

Our analysis shows how tournament prize schedules might be adapted to uncertainty about the environment. Under an entropy bound, robust prizes spread marginal incentive weights as evenly as possible across performance quantiles. In finite tournaments, this produces a unique prize schedule with positive prizes at every rank except the last and a distinct top prize. As the number of agents grows, the schedule approaches the logarithmic profile of the large-tournament model. This profile, together with exponential worst-case noise, also arises under a bound on expected noise above the lower endpoint of its support, and under a broader class of restrictions on means of the quantile hazard.

Robustness, therefore, provides a rationale for broad prize sharing while preserving substantial inequality in rewards. Numerical analysis for $n=3,\ldots,50$ shows both the top prize and the Gini coefficient of robust prizes decline with the number of agents, although Lorenz curves for different values of $n$ intersect multiple times, i.e., they are not ranked by the Lorenz order. The limiting logarithmic schedule has a Gini coefficient of 0.5. Thus, spreading marginal incentive weights evenly does not require equalizing prizes: the logarithmic schedule rewards improvements near the top more strongly than improvements near the bottom. This suggests that the robust prize schedule strikes a balance between high-powered incentives at the top and control over inequality.\footnote{For a perspective, the Gini coefficient for disposable income is between 0.23 in Slovakia and 0.63 in South Africa (\url{https://data.worldbank.org/indicator/SI.POV.GINI}). The average Gini coefficient for compensation in US firms is 0.34 \citep[][based on \emph{Glassdoor} data]{Green-et-al:2025}.} The harmonic prize schedule offers a simple rule of thumb for designing incentives under uncertainty in large finite tournaments.


More broadly, our work contributes to the literature on robust mechanism design when the principal faces uncertainty about the \emph{risk} that agents face. \cite{Carroll-Meng:2016} study this problem for a single agent whose contract is based on output. We study how such uncertainty shapes the allocation of prizes when rewards depend on relative rank. Future research could explore extensions to asymmetric agents or settings where agents themselves face ambiguity.


\phantomsection
\addcontentsline{toc}{section}{References}
\setlength{\bibsep}{0pt}
\bibliographystyle{aea}
\bibliography{robust}

\appendix

\section*{Appendix A. Proofs}

\setcounter{equation}{0}\renewcommand\theequation{A\arabic{equation}}

\setcounter{lemma}{0}
\renewcommand{\thelemma}{A\arabic{lemma}}

\begin{proof}{\bf of Lemma \ref{lemma_equilibrium_cont}}
For given $\gamma$ and $F$, let
\[
R(\delta) = \int w\left(F(t+\delta)\right)\dd F(t), \qquad \delta=x-x^*,
\]
denote expected benefit from an effort deviation $\delta$. Using (\ref{q_def}) and the initial condition $w(0)=0$, we have $w(z) = \int_0^z\frac{\gamma(s)}{1-s}\dd s$. This gives, changing the order of integration,
\begin{align}
\label{Revenue_cont}
& R(\delta) = \int \int_0^{F(t+\delta)}\frac{\gamma(z)}{1-z}\dd z\dd F(t) = \int_0^1\int_{Q(z)-\delta}^{\infty}\dd F(t)\frac{\gamma(z)}{1-z}\dd z \nonumber \\
& = \int_0^1\gamma(z)\frac{1-F(Q(z)-\delta)}{1-z}\dd z.
\end{align}

Suppose first that $f$ is differentiable. Differentiating (\ref{Revenue_cont}) twice gives
\[
R''(\delta) = \int_0^1\gamma(z)\frac{-f'(Q(z)-\delta)}{1-z}\dd z \le M\|h_F\|_{\infty}\int \frac{[-f'(Q(z)-\delta)]_+}{f(Q(z))}\dd z. 
\]
Changing the variable of integration to $y=Q(z)-\delta$ yields
\[
R''(\delta) \le M\|h_F\|_\infty\int_{\mathds R}[-f'(y)]_+\dd y.
\]
Since $f$ is a nonnegative integrable density extended by zero outside its support, its total upward and downward variations coincide. Hence,
$\int_{\mathds R}[-f'(y)]_+\dd y = \frac{1}{2}{\rm TV}(f)$, which gives $R''(\delta)\le
\frac{M}{2}\|h_F\|_\infty{\rm TV}(f)<c_0$.

If $f$ is not differentiable, the same argument applies to a bounded variation density in the distributional sense, with $[-f']_+\dd y$ replaced by the negative-variation measure of $f$. Since $c''\ge c_0$, the payoff in (\ref{payoff_cont}) is strictly concave in own effort. 

We next show that mixed-strategy equilibria are ruled out. Consider an arbitrary distribution $H$ of opponents' efforts, with corresponding output distribution and density 
\[
G(y)=\int F(y-a)\dd H(a),
\qquad
g(y)=\int f(y-a)\dd H(a).
\]
Since $f(y)\le \|h_F\|_\infty[1-F(y)]$, we have
\[
g(y)\le \|h_F\|_\infty\int[1-F(y-a)]\dd H(a)
=\|h_F\|_\infty[1-G(y)],
\]
which implies $\|h_G\|_\infty\le\|h_F\|_\infty$.

Against this distribution, the expected benefit from effort $x$ is, similar to (\ref{Revenue_cont}), 
\[
R_H(x)=\int w(G(x+t))\dd F(t) = \int_0^1\gamma(z)\frac{1-F(Q_G(z)-x)}{1-z}\dd z,
\]
where $Q_G$ is the quantile function of $G$. For differentiable $f$, differentiating twice and changing variables as above then gives
\[
\begin{aligned}
& R_H''(x) = \int_0^1\gamma(z)\frac{-f'(Q_G(z)-x)}{1-z}\dd z \le M\|h_G\|_\infty
\int_{\mathds R}[-f'(y)]_+\dd y\\
&\le \frac{M}{2}\|h_F\|_\infty{\rm TV}(f)<c_0.
\end{aligned}
\]
The same bound holds for bounded variation densities in the distributional sense. Consequently, $R_H(x)-c(x)$ is strictly concave for any distribution of opponents' efforts. Therefore, every agent has a unique pure best response, ruling out mixed-strategy equilibria.

Moreover, since agents are identical and each is negligible, they all face the same aggregate output distribution. Strict concavity then implies that all agents have the same unique best response, ruling out asymmetric equilibria as well.

Finally, since $\gamma(z)>0$ and $h(z)>0$ almost everywhere, $c'(0)=0$, and $c'(x)\ge c_0x$, the unique solution to (\ref{foc_cont_h}) exists and is strictly positive.
\end{proof}

\bigskip

\begin{proof}{\bf of Proposition \ref{prop_continuum}} Consider first $\gamma^*(z)=1$. For every distribution satisfying the entropy restriction, we have
\[
\int_0^1h(z)\dd z \ge \exp\left(\int_0^1\log h(z)\dd z\right) \ge \Lambda^*,
\]
where the first inequality is due to Jensen and the second is the entropy constraint (\ref{entropy_h_cont}). The lower bound is attained by the constant quantile hazard $h(z)=\Lambda^*$, which is the quantile hazard of an exponential distribution with rate $\Lambda^*$, up to a location shift. For this distribution, $\|h_F\|_\infty=\Lambda^*$, ${\rm TV}(f)=2\Lambda^*$, and the condition of the proposition, therefore, implies
\[
\frac{M}{2}\|h_F\|_\infty{\rm TV}(f) = M(\Lambda^*)^2 = Me^{2(1-\bar{H})} < c_0.
\]
Thus, the exponential distribution belongs to $\mathcal R_M$, and hence $\gamma^*$ guarantees marginal incentive $\Lambda^*$ in equilibrium.

Now consider any $\gamma\ne 1$. Since $\gamma$ is not constant almost everywhere, there exists a bounded smooth function $r$, with compact support in $(0,1)$, such that
\[
\int_0^1r(z)\dd z=0 \quad \text{and} \quad \int_0^1\gamma(z)r(z)\dd z<0.
\]
For sufficiently small $t>0$, define
\begin{equation}
\label{perturb_h_cont}
h_t(z) = \Lambda^*a(t)[1+tr(z)], \quad \text{where} \quad a(t) = \exp\left[-\int_0^1\log(1+tr(s))\dd s\right].
\end{equation}
Then
\[
\int_0^1\log h_t(z)\dd z = \log\Lambda^* = 1-\bar H,
\]
so the entropy constraint binds for $h_t$. Moreover, $h_t$ is a smooth perturbation of the constant hazard. Since the exponential distribution satisfies the regularity condition with strict inequality, the distribution induced by $h_t$ belongs to $\mathcal R_M$ for all sufficiently small $t$.

This perturbation produces marginal incentives
\[
\int_0^1\gamma(z)h_t(z)\dd z = \Lambda^*a(t)\left[1+t\int_0^1\gamma(z)r(z)\dd z\right].
\]
Because $\int_0^1r(z)\dd z=0$, $a(0)=1$, and $a'(0)=0$, we have
\[
\left.\frac{\dd}{\dd t}\int_0^1\gamma(z)h_t(z)\dd z\right|_{t=0} = \Lambda^*
\int_0^1\gamma(z)r(z)\dd z < 0.
\]
Thus, every $\gamma\ne 1$ admits a feasible adversarial distribution that generates marginal incentives strictly below $\Lambda^*$. This implies that $\gamma^*(z)=1$ is the unique robust design.

Finally, (\ref{q_def}) implies $w^{*\prime}(z) = \frac{1}{1-z}$. Together with $w^*(0)=0$, this gives (\ref{wstar_cont}). The expression for $c'(x^*)$ follows from Lemma \ref{lemma_equilibrium_cont}.
\end{proof}

\bigskip

\begin{proof}{\bf of Lemma \ref{lemma_equilibrium_regularity}}
Fix an agent $i$ and consider an opponent $j\ne i$, with distribution of effort $G_j$. Opponent $j$'s output has a cdf $K_j$ that is a convolution of $G_j$ with $F$; therefore, $K_j$ is absolutely continuous even if $G_j$ has atoms. The density of $K_j$ is $k_j(y)=\int f(y-x)\dd G_j(x)$, which implies $\|k_j\|_\infty\le\|f\|_\infty$. Let $S_i(y)$ be agent $i$'s expected benefit from output $y$. It is a polynomial in the cdfs $K_j(y)$, $j\ne i$, and thus it is also absolutely continuous.

For each fixed $j\ne i$, we can write this polynomial as

\begin{align}
\label{S_i_through_K_j}
S_i = K_jS_{-i,j}^+ + (1-K_j)S_{-i,j}^-,
\end{align}
where $S_{-i,j}^+$ and $S_{-i,j}^-$ are polynomials that do not depend on $K_j$.\footnote{For example, for $n=3$ and ${\bf v}\in\mathcal{V}$ we have $S_1 = K_2K_3v_1+[K_2(1-K_3)+(1-K_2)K_3]v_2$, which can be written as $S_1=K_2[K_3v_1+(1-K_3)v_2]+(1-K_2)K_3v_2$ or $S_1=K_3[K_2v_1+(1-K_2)v_2]+(1-K_3)K_2v_2$.} Note that $K_j(y)$ is the probability that $Y_j\le y$. Thus, $S_{-i,j}^+$ and $S_{-i,j}^-$ represent, respectively, agent $i$'s expected benefit in scenarios where $j$'s output is below and above $y$. Therefore, if $N_{-i,j}(y)$ denotes a random variable equal to the number of $i$'s opponents other than $j$ whose output exceeds $y$, we have
\[
S_{-i,j}^+ = \mathds{E}[v_{N_{-i,j}(y)+1}] \quad \text{and} \quad S_{-i,j}^- = \mathds{E}[v_{N_{-i,j}(y)+2}].
\]
Indeed, if $N_{-i,j}(y)=m$, agent $i$'s rank is $m+1$ if $j$'s output is below $y$, whereas if $j$'s output is above $y$, agent $i$'s rank is $m+2$. This gives the derivative
\[
\frac{\partial S_i}{\partial K_j} = S_{-i,j}^+ - S_{-i,j}^- = \mathds{E}[v_{N_{-i,j}(y)+1} - v_{N_{-i,j}(y)+2}].
\]
Applying the chain rule to $S_i(y)$ then gives, almost everywhere,
\[
S_i'(y) = \sum_{j\ne i}\frac{\partial S_i}{\partial K_j}k_j(y) = \sum_{j\ne i}k_j(y)\mathds{E}[v_{N_{-i,j}(y)+1} - v_{N_{-i,j}(y)+2}].
\]
Any prize difference $v_m-v_{m+1}$ lies in $[0,v_1-v_n]\subseteq[0,1]$; therefore, $0\le S_i'(y)\le L$, where $L=(n-1)\|f\|_\infty$, and thus $S_i$ is globally Lipschitz.

Agent $i$'s expected benefit from effort $x$ is $R_i(x)=\int S_i(x+t)f(t)\dd t$. Differentiation (using that $S_i$ is Lipschitz and hence its derivative is bounded) yields the marginal benefit
\[
R_i'(x) = \int S_i'(x+t)f(t)\dd t = \int S_i'(y)f(y-x)\dd y.
\]
Therefore, for any $s\ne 0$ we have
\[
|R_i'(x+s)-R_i'(x)|\le L\int|f(y-x-s)-f(y-x)|\dd y
\le L\cdot{\rm TV}(f)|s|<c_0|s|.
\]
It follows that, for $x_2>x_1\ge 0$,
\begin{align*}
& [R_i'(x_2)-c'(x_2)]-[R_i'(x_1)-c'(x_1)] = R_i'(x_2)-R_i'(x_1) - [c'(x_2)-c'(x_1)]\\
& \le |R_i'(x_2)-R_i'(x_1)| - c_0(x_2-x_1) < c_0(x_2-x_1) - c_0(x_2-x_1) = 0.
\end{align*}
Thus, agent $i$'s payoff is strictly concave. Since prizes are bounded and $c(x)\to\infty$ for $x\to \infty$, each agent has a unique best response. Therefore, no equilibrium can involve nondegenerate mixing.

The left-hand side of (\ref{foc_byparts}) is finite and strictly positive. Indeed, $\int f(t)^2\dd t\le\|f\|_\infty$, each $B_r$ in (\ref{B_r_def}), $r=1,\ldots,n-1$, is strictly positive and finite, and at least one prize differential $v_r-v_{r+1}$ is strictly positive. Since $c'(0)=0$, $c'$ is strictly increasing, and $c'(x)\ge c_0x$, there is exactly one positive solution to (\ref{foc_byparts}). Strict concavity makes this solution a global best response to itself. It is, therefore, a symmetric equilibrium, and every symmetric equilibrium must coincide with it.

Finally, we rule out asymmetric equilibria. Let
${\bf x}=(x_1,\ldots,x_n)$ be an arbitrary pure-strategy effort profile, and denote agent $i$'s marginal benefit by 
\[
m_i({\bf x})=\int S_i'(y;{\bf x}_{-i})f(y-x_i)\dd y.
\]
We will now show that
\[
|m_i({\bf x})-m_i(\tilde{\bf x})|
\le (n-1)\|f\|_\infty{\rm TV}(f)
\|{\bf x}-\tilde{\bf x}\|_\infty.
\]
First, holding opponents' efforts ${\bf x}_{-i}$ fixed, we can write
\begin{align*}
&|m_i(x_i,{\bf x}_{-i}) - m_i(\tilde x_i,{\bf x}_{-i})| = \left|\int S_i'(y;{\bf x}_{-i})[f(y-x_i)-f(y-\tilde x_i)]\dd y\right| \\
& = \left|\int \left[S_i'(y;{\bf x}_{-i})-\frac{L}{2}\right][f(y-x_i)-f(y-\tilde x_i)]\dd y\right| \\
& \le \frac{L}{2}\int|f(y-x_i)-f(y-\tilde x_i)|\dd y \le \frac{(n-1)}{2}
\|f\|_\infty{\rm TV}(f)|x_i-\tilde x_i|.
\end{align*}
Here, the second equality follows because $\int[f(y-x_i)-f(y-\tilde x_i)]\dd y = 0$, the first inequality holds because $0\le S_i'\le L$, and and the last inequality is due to $f$ having bounded variation.

Next, consider an opponent $j\ne i$ with effort $x_j$. Since $Y_j=x_j+\epsilon_j$, the probability that $j$'s output is below $y$ is $K_j(y) = F(y-x_j)$; therefore, using (\ref{S_i_through_K_j}),
\[
S_i(y) = F(y-x_j)S^+_{-i,j}(y) + [1-F(y-x_j)]S^-_{-i,j}(y),
\]
where $S^{\pm}_{-i,j}(y)$ depend only on the efforts of opponents other than $i$ and $j$. If $j$ changes his effort to $\tilde{x}_j$, $i$'s expected benefit becomes
\[
\tilde{S}_i(y) = F(y-\tilde{x}_j)S^+_{-i,j}(y) + [1-F(y-\tilde{x}_j)]S^-_{-i,j}(y).
\]
The difference between them, therefore, is
\[
\Delta S_i(y) = S_i(y) - \tilde{S}_i(y) =  
[F(y-x_j)-F(y-\tilde x_j)]
(S_{-i,j}^+(y)-S_{-i,j}^-(y)).
\]
Since $0\le S_{-i,j}^+-S_{-i,j}^-\le 1$,
the function $\Delta S_i(y)$ has a constant sign and satisfies
\[
|\Delta S_i(y)|\le
\|f\|_\infty|x_j-\tilde x_j|.
\]
The corresponding change in marginal benefit is
\begin{align*}
& \Delta m_{i,j} = m(x_i,x_j,{\bf x}_{-i,j}) - m(x_i,\tilde{x}_j,{\bf x}_{-i,j}) \\
& = \int S_i'(y)f(y-x_i)\dd y - \int \tilde{S}_i'(y;)f(y-x_i)\dd y = \int \Delta S_i'(y)f(y-x_i)\dd y.
\end{align*}
Integration by parts then gives
\[
\Delta m_{i,j}=-\int\Delta S_i(y)\dd f(y-x_i).
\]
The signed measure $\dd f$ has total mass zero, and its positive and negative variations each have mass ${\rm TV}(f)/2$. Since $\Delta S_i$ has a constant sign, it follows that
\[
|\Delta m_{i,j}|\le \frac{1}{2}\|f\|_\infty{\rm TV}(f)|x_j-\tilde x_j|.
\]

Then we can write
\[
m_i({\bf x})-m_i(\tilde{\bf x}) = m_i(x_i,{\bf x}_{-i})-m_i(\tilde{x}_i,{\bf x}_{-i}) + \sum_{j\ne i}\Delta m_{i,j},
\]
where the sum is understood as a telescoping sum, with agent $i$'s effort fixed at $\tilde x_i$ and opponents' efforts changed successively from $x_j$ to $\tilde x_j$. This  gives
\begin{align*}
& |m_i({\bf x})-m_i(\tilde{\bf x})|
\le \frac{\|f\|_\infty{\rm TV}(f)}{2}
\left[(n-1)|x_i-\tilde x_i|
+\sum_{j\ne i}|x_j-\tilde x_j|\right]\\
& \le (n-1)\|f\|_\infty{\rm TV}(f)
\|{\bf x}-\tilde{\bf x}\|_\infty.
\end{align*}

Now define the mapping $\Phi:\mathbb R_+^n\to\mathbb R_+^n$ by $\Phi_i({\bf x})=c^{\prime{-1}}(m_i({\bf x}))$. Since $c''\ge c_0$, $c^{\prime{-1}}$ is Lipschitz with
constant $1/c_0$. Therefore,
\[
\|\Phi({\bf x})-\Phi(\tilde{\bf x})\|_\infty
\le \frac{(n-1)\|f\|_\infty{\rm TV}(f)}{c_0}
\|{\bf x}-\tilde{\bf x}\|_\infty.
\]
By the definition of $\mathcal R_n$, the coefficient on the right-hand side is strictly less than one. Thus, $\Phi$ is a contraction and has a unique fixed point. Strict concavity of the payoff function implies that its fixed points are precisely the pure-strategy
equilibria. Since the game is symmetric, any permutation of an equilibrium effort profile is also an equilibrium. Uniqueness therefore implies that the equilibrium is symmetric,
ruling out asymmetric equilibria.
\end{proof}

\bigskip

\begin{proof}{\bf of Lemma \ref{lemma h(z)}}
For the objective in problem (\ref{problem_min}), we can write
\begin{align*}
& \int_0^1 h(z)g(z;{\bf p})\dd z = \exp\left[\log\int_0^1 h(z)g(z;{\bf p})\dd z\right] \ge \exp\left[\int_0^1\log[h(z)g(z;{\bf p})]\dd z\right]\\
& = \exp\left[\int_0^1\log h(z) \dd z + \int_0^1\log g(z;{\bf p})\dd z\right] \ge \exp\left[1-\bar{H} + \int_0^1\log g(z;{\bf p})\dd z\right].
\end{align*}
Here, the first inequality is due to Jensen, and the second one is due to the entropy constraint in (\ref{problem_min}). The final expression is independent of $h$, and hence provides the optimal objective if it is attainable. Attainment requires both the entropy constraint to bind and Jensen's inequality to hold with equality. The latter occurs if and only if there exists a constant $\lambda({\bf p})>0$ such that $h(z)g(z;{\bf p})=\lambda({\bf p})$ almost everywhere. Together with (\ref{entropy_h_cont}), this yields (\ref{h*_p}) and uniquely determines the optimizer almost everywhere.

To show that \(F^*(\cdot;{\bf p})\) has bounded support, note that
\begin{equation}
\label{Q'(z)}
Q'(z)=\frac{1}{f(Q(z))}=\frac{1}{(1-z)h(z)}.
\end{equation}
Under our interval-support convention, the length of the support is, therefore,
\[
\lim_{z\to 1}Q(z)-\lim_{z\to 0}Q(z) = \int_0^1\frac{\dd z}{(1-z)h(z)}.
\]
Hence, the support is bounded if and only if $\frac{1}{(1-z)h(z)}$ is integrable on $(0,1)$. For solution (\ref{h*_p}), we have
\begin{align*}
& \int_0^1\frac{\dd z}{(1-z)h^*(z;{\bf p})}
= \frac{1}{\lambda({\bf p})}\int_0^1\frac{g(z;{\bf p})}{1-z}\dd z
= \frac{1}{\lambda({\bf p})}\sum_{r=1}^{n-1}p_r\int_0^1\frac{g_r(z)}{1-z}\dd z \\
& = \frac{1}{\lambda({\bf p})}\sum_{r=1}^{n-1}p_r
\frac{n!}{(n-1-r)!r!}
\int_0^1z^{n-r-1}(1-z)^{r-1}\dd z \\
& = \frac{1}{\lambda({\bf p})}\sum_{r=1}^{n-1}p_r\frac{n}{r}
= \frac{n}{\lambda({\bf p})}\sum_{r=1}^{n-1}(v_r-v_{r+1})
= \frac{nv_1}{\lambda({\bf p})}<\infty.
\end{align*}
Thus, the corresponding interval-support distribution $F^*(\cdot;{\bf p})$ has bounded support. This calculation also verifies that $h^*(\cdot;{\bf p})$ is admissible.
\end{proof}

\bigskip

\begin{proof}{\bf of Proposition \ref{prop_main_robust}}
We start by establishing existence. For any ${\bf p}\in\Delta^{n-1}$, at least one component satisfies $p_r\ge \frac{1}{n-1}$, and hence
\[
\frac{1}{n-1}\min_{1\le r\le n-1}g_r(z) \le g(z;{\bf p}) \le \sum_{r=1}^{n-1}g_r(z).
\]
Since $g_r(z)\propto z^{n-r-1}(1-z)^r$, the logarithms of both bounds are integrable, giving integrable bounds for $|\log g(z;{\bf p})|$, uniformly in ${\bf p}$. Dominated convergence then implies that $\int_0^1\log g(z;{\bf p})\dd z$ is continuous on $\Delta^{n-1}$, and compactness of the simplex guarantees existence of a maximizer. Moreover, the objective is strictly concave because $\log$ is strictly concave and the densities $g_1,\ldots,g_{n-1}$ are linearly independent. Hence, the maximizer is unique.

The derivatives of the objective in (\ref{problem_max}) with respect to $p_r$ are
\[
\ell_r({\bf p}) = \int_0^1\frac{g_r(z)}{g(z;{\bf p})}\dd z, \qquad r=1,\ldots,n-1.
\]
The Kuhn--Tucker conditions, therefore, are
\[
\ell_r({\bf p})\le \mu, \qquad\text{with equality if }p_r>0, \qquad r=1,\ldots,n-1,
\]
for some $\mu>0$. Multiplying these conditions by $p_r$ and summing over $r$ gives
\[
\mu = \sum_{r=1}^{n-1}p_r\ell_r({\bf p}) = \int_0^1\frac{\sum_{r=1}^{n-1}p_rg_r(z)}
{g(z;{\bf p})}\dd z =1.
\]
[Note that conditions that hold as inequalities have $p_r=0$ and do not contribute to the sum.] Thus, at the optimum $\mu^*=1$ and
\begin{align}
\label{KT_conditions}
\ell_r({\bf p}^*)\le 1, \qquad\text{with equality if }p_r^*>0. \qquad r=1,\ldots,n-1.
\end{align}

We now show that $p_1^*>0$. Suppose instead that $p_1^*=0$, and let $r_1=\min\{r:p_r^*>0\}\ge 2$. As $z\to 1$, the leading behavior of $g(z;{\bf p}^*)$ is $g(z;{\bf p}^*)\propto(1-z)^{r_1}$. Since $g_1(z)\propto 1-z$, we have
\[
\frac{g_1(z)}{g(z;{\bf p}^*)}\propto (1-z)^{1-r_1},
\]
which is not integrable at $z=1$ for $r_1\ge 2$, and hence $\ell_1({\bf p}^*)=\infty$, which contradicts the optimality condition $\ell_1({\bf p}^*)\le 1$. Therefore, it must be that $p_1^*>0$.

Similarly, suppose that $p_{n-1}^*=0$, and let $r_2=\max\{r:p_r^*>0\}\le n-2$. As $z\to 0$, the leading behavior of $g(z;{\bf p}^*)$ is $g(z;{\bf p}^*)\propto z^{n-r_2-1}$, whereas $g_{n-1}(z)$ converges to a strictly positive constant. Therefore,
\[
\frac{g_{n-1}(z)}{g(z;{\bf p}^*)} \propto z^{-(n-r_2-1)},
\]
which is not integrable at $z=0$ for $n-r_2-1\ge 1$, and hence $\ell_{n-1}({\bf p}^*)=\infty$, again contradicting the optimality condition. Hence, $p_{n-1}^*>0$.

Finally, recall that $v_r-v_{r+1}=\frac{p_r^*}{r}$; therefore, $p_1^*>0$ implies $v_1>v_2$ (a distinct top prize), while $p_{n-1}^*>0$ implies $v_{n-1}>v_n=0$ and hence $v_r>0$ for all $r=1,\ldots,n-1$.
\end{proof}

\bigskip

\begin{proof}{\bf of Lemma \ref{lemma_implementation}}
Fix $n\ge 2$. By Proposition \ref{prop_main_robust}, the polynomial $k(z)=\frac{g(z;{\bf p}^*)}{1-z}$ is strictly positive on $[0,1]$. Indeed, it is positive in $(0,1)$, with $k(0)=np_{n-1}^*>0$ and $k(1)=n(n-1)p_{1}^*>0$. Let $F_{\bar{H}}^*$ denote its interval-support adversarial distribution at entropy bound $\bar{H}$. Writing $\lambda_{\bar{H}}=e^{1-\bar{H}+J({\bf p}^*)}$, the quantile density of $F_{\bar{H}}^*$ is $(Q_{\bar{H}}^*)'(z)=\frac{k(z)}{\lambda_{\bar{H}}}$, and hence $f_{\bar{H}}^*(Q_{\bar{H}}^*(z))=\frac{\lambda_{\bar{H}}}{k(z)}$. This gives
\begin{align}
\label{sup_TV_proof}
\|f_{\bar{H}}^*\|_\infty = \frac{\lambda_{\bar{H}}}{\min_{z\in[0,1]}k(z)}, \quad {\rm TV}(f_{\bar{H}}^*) = \lambda_{\bar{H}}\left[\frac1{k(0)}+\frac1{k(1)} + \int_0^1\frac{|k'(z)|}{k(z)^2}\dd z\right]<\infty.
\end{align}
The first two terms in total variation are the jumps at the endpoints of the support. 

At a finite reference entropy bound $\bar{H}_0$, define
\[
C_n = (n-1)\|f_{\bar{H}_0}^*\|_\infty{\rm TV}(f_{\bar{H}_0}^*),
\qquad \bar{H}_n = \bar{H}_0 + \frac{1}{2}\log\frac{C_n}{c_0}.
\]
Both are finite, and $C_n>0$. Now consider the dilation of $f_{\bar{H}_0}^*$ with $\bar{H}>\bar{H}_0$ such that $\bar{H}>\bar{H}_n$. Then
\[
(n-1)\|f_{\bar H}^*\|_\infty{\rm TV}(f_{\bar H}^*) = C_ne^{-2(\bar H-\bar{H}_0)} = c_0e^{-2(\bar{H}-\bar{H}_n)}<c_0.
\]
The entropy of $f_{\bar{H}}^*$ is $\bar H$, so $F_{\bar H}^*\in\mathcal{F}_n$. This proves part (ii) of the lemma.

For part (iii), and to establish the optimality of ${\bf v}^*$ in part (i), let $g^*(z)=g(z;{\bf p}^*)$ and $\lambda^*=e^{1-\bar H+J({\bf p}^*)}$. From the proof of Proposition (\ref{prop_main_robust}), the Kuhn--Tucker conditions for (\ref{problem_max}) are
\[
\int_0^1\frac{g_r(z)}{g^*(z)}\dd z\le1,
\quad r=1,\ldots,n-1,
\]
with equality for $p^*_r>0$. Thus, for any ${\bf p}\in\Delta^{n-1}$,
\[
\int_0^1h^*(z;{\bf p}^*)g(z;{\bf p})\dd z = \lambda^*\int_0^1\frac{g(z;{\bf p})}{g^*(z)}\dd z = \lambda^*\sum_{r=1}^{n-1}p_r\int_0^1\frac{g_r(z)}{g^*(z)}\dd z \le \lambda^*\sum_{r=1}^{n-1}p_r = \lambda^*.
\]
For every $F\in\mathcal F_n(\bar H)$ and associated $h$, the entropy inequality in Lemma \ref{lemma h(z)} gives $\int_0^1h(z)g^*(z)\dd z\ge\lambda^*$. Lemma \ref{lemma_equilibrium_regularity} and monotonicity of $(c')^{-1}$ then establish the value (\ref{implemented_value}), optimality of ${\bf p}^*$ and ${\bf v}^*$, and the saddle property.

To prove the uniqueness of ${\bf v}^*$, fix ${\bf p}\ne{\bf p}^*$. If $\int_0^1h^*(z;{\bf p}^*)g(z;{\bf p})\dd z<\lambda^*$, $F_{\bar{H}}^*$ already gives this schedule a strictly smaller effort guarantee. Assuming a tie, $\int_0^1h^*(z;{\bf p}^*)g(z;{\bf p})\dd z=\lambda^*$, define
\[
r(z)=\frac{g(z;{\bf p})-g^*(z)}{g^*(z)} \quad \text{and} \quad {\bf p}_t=(1-t){\bf p}^*+t{\bf p}, \quad t\in[0,1].
\]
The tie implies $\int_0^1r(z)\dd z=0$. The ratio $r(z)$ extends continuously to $[0,1]$, since the common factor $(1-z)$ in the numerator and denominator cancels out leaving strictly positive polynomials. Also, $r(z)$ is not identically zero, by linear independence of the Bernstein basis components $g_r$. Note that $k_t(z)=\frac{g(z;{\bf p}_t)}{1-z}$ is $C^1$ in $t$; therefore, for small $t\ge 0$ it is strictly positive and converges to $k(z)$ in $C^1[0,1]$. Equations (\ref{sup_TV_proof}) above show that the interval-support adversarial distributions $F_t=F^*(\cdot;{\bf p}_t)$ have $(n-1)\|f_t\|_\infty{\rm TV}(f_t)$ converging to that of $F_{\bar{H}}^*$. The strict inequality $(n-1)\|f_{\bar{H}}^*\|_\infty{\rm TV}(f_{\bar{H}}^*)<c_0$, therefore, ensures that $F_t\in\mathcal F_n(\bar H)$ for all sufficiently small $t$.

By construction, we have 
\begin{align}
\label{gg}
g(z;{\bf p})=g^*(z)(1+r(z)), \quad g(z;{\bf p}_t) = (1-t)g^*(z) + tg(z;{\bf p}) = (1+tr(z))g^*(z).
\end{align}
The adversarial hazard corresponding to $F_t$ is $h_t(z) = \frac{\lambda({\bf p}_t)}{g(z;{\bf p}_t)}$; therefore, the marginal benefit under noise $F_t$ and prizes ${\bf p}$ is
\[
L(t) = \frac{1}{n}\int_0^1h_t(z)g(z;{\bf p})\dd z = \frac{\lambda({\bf p}_t)}{n}\int_0^1\frac{g(z;{\bf p})}{g(z;{\bf p}_t)}\dd z = \frac{\lambda({\bf p}_t)}{n}\int_0^1\frac{1+r(z)}{1+tr(z)}\dd z.
\]
Recall that 
\[
\lambda({\bf p}_t) = \exp\left[1-\bar{H}+\int_0^1\log g(z;{\bf p}_t)\dd z\right],
\]
which we can rewrite using (\ref{gg}) as
\[
\lambda({\bf p}_t) = \exp\left[1-\bar{H}+\int_0^1[\log g^*(z) + \log(1+tr(z))]\dd z\right] = \lambda^*\exp\left[\int_0^1\log(1+tr(z))\dd z\right].
\]
Differentiating with respect to $t$ (which is justified because $r(z)$ is bounded and $1+tr(z)>0$ for sufficiently small $t$), we obtain
\[
\frac{\dd}{\dd t}\lambda({\bf p}_t)\bigg|_{t=0} = \lambda^*\int_0^1r(z)\dd z = 0,
\]
where the last equality follows from the tie assumption. Thus,
\[
L(0)=\frac{\lambda^*}{n} \quad \text{and} \quad L'(0) = -\frac{\lambda^*}{n}\int_0^1r(z)^2\dd z<0.
\]
Thus, for any ${\bf p}\ne{\bf p}^*$, there is a feasible perturbation $F_t$ of $F^*_{\bar{H}}$ that gives a strictly smaller equilibrium effort than $x^*[{\bf v}^*,F^*_{\bar{H}}]$. Therefore, for any ${\bf v}\ne{\bf v}^*$, the effort guarantee is strictly below $x^*[{\bf v}^*,F^*_{\bar{H}}]$, which proves unique optimality.
\end{proof}

\bigskip

\begin{proof}{\bf of Proposition \ref{prop_asymptotic}}
We first establish convergence of the optimal cutoff density, $g(z;{\bf p}^{*,n})$, to 1 in $L^1$ and uniform convergence of $G(z;{\bf p}^{*,n})$ to the uniform distribution. For equal budget shares,
\begin{align}
\label{g_equal_budget}
g(z;\bar{\bf p}^n) = \frac{1}{n-1}\sum_{r=1}^{n-1}n\binom{n-1}{r}z^{n-r-1}(1-z)^r =
\frac{n}{n-1}\left(1-z^{n-1}\right).
\end{align}
Therefore, the objective converges to zero:
\begin{align*}
J(\bar{\bf p}^n) = \log\frac{n}{n-1} + \int_0^1\log(1-z^{n-1})\dd z = \log\frac{n}{n-1} - \sum_{k=1}^{\infty}\frac{1}{k((n-1)k+1)} \longrightarrow 0.
\end{align*}
Since $g(\cdot;{\bf p})$ is a probability density, Jensen's inequality implies, for any ${\bf p}$,
\[
J({\bf p}) = \int_0^1\log g(z;{\bf p})\dd z \le \log\int_0^1g(z;{\bf p})\dd z = 0.
\]
This gives, due to optimality of ${\bf p}^{*,n}$,
\[
0 \ge J({\bf p}^{*,n}) \ge J(\bar{\bf p}^n) \longrightarrow 0,
\]
and hence
\[
{\rm D}_{\rm KL}(U_{[0,1]}\| G(\cdot;{\bf p}^{*,n})) = -J({\bf p}^{*,n}) \longrightarrow 0. 
\]
Pinsker's inequality \citep[e.g.,][Lemma 2.5(i)]{Tsybakov-Zaiats:2009} then yields the $L^1$ convergence:
\begin{equation}
\label{g_uniform_L1}
\int_0^1\left|g(z;{\bf p}^{*,n})-1\right|\dd z \longrightarrow 0.
\end{equation}
For brevity, let $G_n(z)=G(z;{\bf p}^{*,n})$. Using $G_n(z)-z = \int_0^z(g(t;{\bf p}^{*,n})-1)\dd t$, (\ref{g_uniform_L1}) then implies uniform convergence:
\begin{equation}
\label{G_uniform}
\sup_{z\in[0,1]}|G_n(z)-z| = \sup_{z\in[0,1]}\left|\int_0^z(g(t;{\bf p}^{*,n})-1)\dd t\right|\le \int_0^1|g(t;{\bf p}^{*,n})-1|\dd t \longrightarrow 0.
\end{equation}

\noindent{\bf Part (i)}: Here, we use the Bernstein polynomial structure of $G_n(z)$ and show that, because $G_n(z)$ converges uniformly to the uniform distribution, \emph{cumulative} budget shares $A_{j,n} = \sum_{r=n-j}^{n-1}p_r^{*,n}$ converge  uniformly to $\frac{j}{n}$. Define cumulative weights $A_{j,n}$ as follows:
\[
A_{0,n}=0, \qquad A_{j,n} = \sum_{r=n-j}^{n-1}p_r^{*,n}, \quad j=1,\ldots,n-1,
\qquad A_{n,n}=1.
\]
Then
\begin{equation}
\label{G_Bernstein}
G_n(z) = \sum_{j=0}^n A_{j,n}\binom{n}{j}z^j(1-z)^{n-j}.
\end{equation}
Indeed, differentiating (\ref{G_Bernstein}) gives
\[
G_n'(z) = n\sum_{j=0}^{n-1}(A_{j+1,n}-A_{j,n})\binom{n-1}{j}z^j(1-z)^{n-1-j} = g(z;{\bf p}^{*,n}),
\]
because $A_{j+1,n}-A_{j,n} = p_{n-1-j}^{*,n}$ for $j=0,\ldots,n-2$ and $A_{n,n}-A_{n-1,n}=0$, and $G_n(0)=0$.

We will show that
\begin{equation}
\label{A_uniform}
\max_{0\leq j\leq n}\left|A_{j,n}-\frac{j}{n}\right|\longrightarrow 0.
\end{equation}
Let $X_{n,z}\sim{\rm Binomial}(n,z)$, $\delta_n=n^{-1/4}$, and consider indices $j$ such that $\frac{j}{n}\in[\delta_n,1-\delta_n]$. Define $z_{\pm}=\frac{j}{n}\pm\delta_n\in[0,1]$. Then we can write
\begin{align}
\label{Hoeffding}
& \mathds{P}[X_{n,z_+}<j] \le \mathds{P}[X_{n,z_+}-nz_+\le j-nz_+] = \mathds{P}[X_{n,z_+}-nz_+ \le -n\delta_n] \le e^{-2n\delta_n^2},\nonumber\\
& \mathds{P}[X_{n,z_-}>j] \le \mathds{P}[X_{n,z_-}-nz_-\ge j-nz_-] = \mathds{P}[X_{n,z_-}-nz_- \ge n\delta_n] \le e^{-2n\delta_n^2},
\end{align}
where, in both cases, the first inequality holds because we add the event $\{X_{n,z_{\pm}}=j\}$, and the last step is due to Hoeffding's inequality \citep[][Theorem 1]{Hoeffding:1963}.\footnote{Recall that $X_{n,z}=\sum_{i=1}^nY_i$, where $Y_i$ are i.i.d. Bernoulli with parameter $z$. Hoeffding's inequality then gives $\mathds{P}[X_{n,z}-nz\le -nt]\le e^{-2nt^2}$ and $\mathds{P}[X_{n,z}-nz\ge nt]\le e^{-2nt^2}$ for $t\ge 0$.}

Note that (\ref{G_Bernstein}) can be written as $G_n(z)=\mathds{E}[A_{X_{n,z},n}]$. Since $A_{j,n}\in[0,1]$ and they are nondecreasing in $j$, we have
\begin{align*}
& G_n(z_+) = \mathds{E}[A_{X_{n,z_+},n}]\ge \mathds{E}[A_{X_{n,z_+},n}\mathds{1}\{X_{n,z_+}\ge j\}] \ge A_{j,n}\mathds{P}[X_{n,z_+}\ge j], \\
& G_n(z_-) = \mathds{E}[A_{X_{n,z_-},n}] = \mathds{E}[A_{X_{n,z_-},n}\mathds{1}\{X_{n,z_-}\le j\}] + \mathds{E}[A_{X_{n,z_-},n}\mathds{1}\{X_{n,z_-}> j\}] \\
& \le A_{j,n}\mathds{P}[X_{n,z_-}\le j] + \mathds{P}[X_{n,z_-}>j] \le A_{j,n} + \mathds{P}[X_{n,z_-}>j].
\end{align*}
Using (\ref{G_uniform}) and (\ref{Hoeffding}), these inequalities imply, uniformly over such $j$,
\[
A_{j,n}\le \frac{G_n(z_+)}{1-e^{-2n\delta_n^2}} = \frac{z_+ + o(1)}{1-e^{-2n\delta_n^2}} = z_+ + o(1) = \frac{j}{n} + \delta_n + o(1)
\]
and
\[
A_{j,n}\ge G_n(z_-) - e^{-2n\delta_n^2} = z_- - e^{-2n\delta_n^2} - o(1) = z_- - o(1) = \frac{j}{n} - \delta_n - o(1).
\]
Thus,
\[
\frac{j}{n} - \delta_n - o(1) \le A_{j,n} \le \frac{j}{n} + \delta_n + o(1).
\]

These bounds are obtained over interior $j$ such that $\frac{j}{n}\in[\delta_n,1-\delta_n]$. Consider now indices $j$ such that $\frac{j}{n}$ is within $\delta_n$ of either endpoint. For $\frac{j}{n}<\delta_n$, let $j_n^- = \lceil n\delta_n\rceil$. Then $\frac{j_n^-}{n}=\delta_n+O(n^{-1})$, so the interior upper bound gives
\[
A_{j_n^-,n}\le \frac{j_n^-}{n} + \delta_n + o(1) = 2\delta_n + O(n^{-1}) + o(1)
\]
Because $\frac{j}{n}\le \delta_n$, this implies
\[
A_{j,n} - \frac{j}{n} \le A_{j_n^-,n} - \frac{j}{n} \le A_{j_n^-,n} \le 2\delta_n + O(n^{-1}) + o(1)
\]
and
\[
\frac{j}{n} - A_{j,n} \le \frac{j}{n} \le \delta_n.
\]
Thus, uniformly over $\frac{j}{n}\le \delta_n$ we have
\[
\frac{j}{n} - \delta_n \le A_{j,n} \le \frac{j}{n} + 2\delta_n + O(n^{-1}) + o(1).
\]

Similarly, consider $\frac{j}{n}>1-\delta_n$ and let $j_n^+ = \lfloor n(1-\delta_n)\rfloor$. Then $\frac{j_n^+}{n} = 1-\delta_n+O(n^{-1})$, and the interior lower bound gives
\[
A_{j_n^+,n} \ge \frac{j_n^+}{n} - \delta_n - o(1) \ge 1 - 2\delta_n - O(n^{-1}) - o(1).
\]
If $\frac{j}{n}\ge 1-\delta_n$, then $j\ge j_n^+$ and
\[
1 \ge A_{j,n} \ge A_{j_n^+,n} \ge 1 - 2\delta_n - O(n^{-1}) - o(1).
\]
This gives
\[
A_{j,n} - \frac{j}{n} \le 1 - \frac{j}{n} \le \delta_n
\]
and
\[
\frac{j}{n} - A_{j,n} \le 1 - A_{j,n} \le 2\delta_n + O(n^{-1}) + o(1).
\]
Thus, uniformly over $\frac{j}{n}\ge 1-\delta_n$ we have
\[
\frac{j}{n} - 2\delta_n - O(n^{-1}) - o(1) \le A_{j,n} \le \frac{j}{n} + \delta_n.
\]
Combining the bounds for the endpoint regions and the interior, we obtain
\[
\left|A_{j,n}-\frac{j}{n}\right|\le 2\delta_n + O(n^{-1}) + o(1) \longrightarrow 0
\]
uniformly over $j=0,1,\ldots,n$. This proves (\ref{A_uniform}).

Now we can analyze the asymptotic behavior of optimal prizes. From the definition of coefficients $A_{j,n}$, we have
\[
\sum_{r=1}^m p_r^{*,n} = 1 - \sum_{r=m+1}^{n-1}p_r^{*,n} = 1 - A_{n-1-m,n}, \quad m=1,\ldots,n-1.
\]
Note also that 
\[
1-\frac{n-1-m}{n}-\frac{m}{n-1} = \frac{n(n-1)-(n-1)(n-1-m)-mn}{n(n-1)} = \frac{n-1-m}{n(n-1)}\in\left[0,\frac{1}{n}\right],
\]
and hence, using the triangle inequality,
\[
\left|\sum_{r=1}^m p_r^{*,n}-\frac{m}{n-1}\right| = \left|1-A_{n-1-m,n}-\frac{m}{n-1}\right|\le \left|A_{n-1-m,n}-\frac{n-1-m}{n}\right| + \frac{1}{n}.
\]
The right-hand side converges to zero uniformly over $m$ due to (\ref{A_uniform}); therefore,
\begin{align}
\label{uniform_cumulative_shares}
\max_{1\le m\le n-1}\left|\sum_{r=1}^m p_r^{*,n}-\frac{m}{n-1}\right|\longrightarrow 0.
\end{align}

We now show that uniform convergence of cumulative budget shares (\ref{uniform_cumulative_shares}) implies weak convergence of the corresponding distribution of budget shares over normalized ranks to $U_{[0,1]}$. Define probability measure
\[
\nu_n = \sum_{k=1}^{n-1}p_k^{*,n}\delta_{k/(n-1)}.
\]
We can show that its cdf, $K_n(z) = \sum_{k:k/(n-1)\le z}p_k^{*,n}$, converges uniformly to $z$. Indeed, for arbitrary $z\in[0,1]$ let $m=\lfloor(n-1)z\rfloor$. Then $K_n(z) = \sum_{k=1}^mp_k^{*,n}$ and
\[
|K_n(z)-z|\le \left|\sum_{k=1}^mp_k^{*,n}-\frac{m}{n-1}\right| + \left|\frac{m}{n-1}-z\right|\le \max_{m=1,\ldots,n-1}\left|\sum_{k=1}^mp_k^{*,n}-\frac{m}{n-1}\right| + \frac{1}{n-1}.
\]
Thus, $\sup_{z\in[0,1]}|K_n(z)-z|\longrightarrow 0$ from (\ref{uniform_cumulative_shares}). Because $K(z)=z$ is the cdf of $U_{[0,1]}$, we conclude that $\nu_n$ converges weakly to $U_{[0,1]}$.

Consider a sequence of indices $r_n$ such that $a_n=\frac{r_n}{n-1}\to\alpha\in(0,1)$. 
Since $v_{r_n}^* = \sum_{k=r_n}^{n-1}\frac{p_k^{*,n}}{k}$, we can write
\[
(n-1)v_{r_n}^{*,n} = \sum_{k=r_n}^{n-1}\frac{p_k^{*,n}}{k/(n-1)} = \int_{[a_n,1]}\frac{1}{z}\dd\nu_n(z) = \int_{[\alpha,1]}\frac{1}{z}\dd\nu_n(z) + \xi_n.
\]
Weak convergence gives
\[
\int_{[\alpha,1]}\frac{1}{z}\dd\nu_n(z) \longrightarrow \int_{\alpha}^1\frac{1}{z}\dd z = -\log \alpha,
\]
because function $\frac{1}{z}\mathds{1}\{z\ge\alpha\}$ is bounded, and its only discontinuity is at $\alpha$ where the uniform distribution has no atom. The remainder satisfies
\[
|\xi_n| = \left|\int_{[a_n,1]}\frac{1}{z}\dd\nu_n(z) - \int_{[\alpha,1]}\frac{1}{z}\dd\nu_n(z)\right|\le \frac{2}{\alpha}\nu_n([\min\{a_n,\alpha\},\max\{a_n,\alpha\}])
\]
for all $n$ large enough such that $a_n\ge\frac{\alpha}{2}$. The mass of the shrinking interval converges to zero. Thus, $(n-1)v_{r_n}^{*,n}\longrightarrow -\log\alpha$.

For the harmonic prize schedule (\ref{v_r theory}), $(n-1)\bar{v}_{r_n}^n = \sum_{k=r_n}^{n-1}\frac{1}{k} \longrightarrow -\log\alpha$ as well, which proves part (i).

\noindent{\bf Part (ii)}: Let $h_n^*$ denote the adversarial quantile hazard associated with ${\bf p}^{*,n}$. By Lemma \ref{lemma h(z)},
\[
h_n^*(z) = \frac{\lambda_n^*}{g(z;{\bf p}^{*,n})}, \qquad \lambda_n^* = \exp\left[1-\bar H+J({\bf p}^{*,n})\right].
\]
We showed in the proof of part (i) that $J({\bf p}^{*,n})\to 0$; therefore, $\lambda_n^*
\longrightarrow e^{1-\bar H}$.

From (\ref{Q'(z)}), we have 
\[
Q_n^*(z) = \frac{1}{\lambda_n^*}\int_0^z\frac{g(s;{\bf p}^{*,n})}{1-s}\dd s.
\]
For every fixed $z<1$, function $\frac{1}{1-s}$ is bounded on $[0,z]$, so (\ref{g_uniform_L1}) gives
\[
\int_0^z \frac{g(s;{\bf p}^{*,n})}{1-s}\dd s \longrightarrow \int_0^z\frac{1}{1-s}\dd s
= -\log(1-z).
\]
Therefore, $Q_n^*(z)\longrightarrow -e^{\bar H-1}\log(1-z)$ for all $z\in(0,1)$. This is the quantile function of the exponential distribution with rate $e^{1-\bar H}$. Weak convergence of the corresponding distributions follows because the limiting quantile function is continuous and strictly increasing.
\end{proof}

\newpage

\section*{Appendix B (online): Additional results}
\label{app_technical}

\setcounter{section}{2}

\setcounter{lemma}{0}
\renewcommand{\thelemma}{B\arabic{lemma}}

\subsection{The minimax result for problem (\ref{problem_maxmin})}
\label{app_Sion}

In this section, we show that maximization over prize schedules and minimization over admissible quantile hazards in problem 
\[
\sup_{{\bf p}\in\Delta^{n-1}}\inf_{h\in \mathcal{H}}\int_0^1h(z)g(z;{\bf p})\dd z
\]
can be interchanged. Once the optimal hazard is shown to be implementable by a distribution $F^*\in\mathcal F_n(\bar H)$, the same minimax equality holds for the original robust-design problem (\ref{problem_maxmin}).

Recall that $g(z;{\bf p})=\sum_{r=1}^{n-1}p_rg_r(z)$, where each $g_r$ is a probability density on $[0,1]$, and let $\mathcal H$ denote the set of admissible quantile hazards satisfying $\int_0^1\log h(z)\dd z\ge 1-\bar H$. 
Here, admissibility requires $h>0$ almost everywhere, $\log h\in L^1(0,1)$, and that $h$ generates a well-defined quantile function solving
\[
Q'(z)=\frac{1}{(1-z)h(z)}, \qquad z\in(0,1).
\]

Define $\bar g(z)=\sum_{r=1}^{n-1}g_r(z)$ and the weighted space 
\[
L^1_{\bar g} = \left\{h:\int_0^1|h(z)|\bar g(z)\dd z<\infty\right\},
\qquad \|h\|_{\bar g} = \int_0^1|h(z)|\bar g(z)\dd z.
\]
We will apply the minimax theorem to the subset
\[
\mathcal H_f = \mathcal H\cap L^1_{\bar g} = \left\{h\in\mathcal H:
\int_0^1h(z)g_r(z)\dd z<\infty \text{ for all }r=1,\ldots,n-1 \right\}.
\]

\begin{lemma}
\label{lemma_minimax}
The value of Nature's problem is unchanged by restricting attention to $\mathcal H_f$. Moreover,
\[
\sup_{{\bf p}\in\Delta^{n-1}}\inf_{h\in\mathcal H}\int_0^1h(z)g(z;{\bf p})\dd z
= \inf_{h\in\mathcal H_f}\sup_{{\bf p}\in\Delta^{n-1}}\int_0^1h(z)g(z;{\bf p})\dd z.
\]
\end{lemma}

\begin{proof}
We first show that restricting Nature to \(\mathcal H_f\) does not change the inner infimum. Fix \({\bf p}\in\Delta^{n-1}\). Since a constant hazard satisfying the entropy constraint is feasible and yields a finite objective, hazards for which $\int_0^1h(z)g(z;{\bf p})\dd z=\infty$ are irrelevant for the infimum. Consider, therefore, any $h\in\mathcal H$ for which this integral is finite.

For $K>0$, define the upper truncation $\tilde h_K(z)=\min\{h(z),K\}$. Since $\tilde h_K$ is bounded and each $g_r$ integrates to one,
\[
\int_0^1\tilde h_K(z)g_r(z)\dd z\le K<\infty \qquad \text{for all }r=1,\ldots,n-1.
\]
Upper truncation may violate the entropy constraint. Let $C=1-\bar H$ and define
\[
c_K = \exp\left\{\left[C-\int_0^1\log\tilde h_K(z)\dd z\right]_+\right\},
\qquad h_K=c_K\tilde h_K.
\]
Then $h_K$ satisfies the entropy constraint by construction and belongs to $\mathcal H_f$. Upper truncation and multiplication by a positive constant also preserve admissibility as a quantile hazard.

Because $\log h\in L^1(0,1)$,
\[
\int_0^1\log\tilde h_K(z)\dd z \longrightarrow \int_0^1\log h(z)\dd z
\ge C,
\]
and hence $c_K\to 1$. Moreover, $\tilde h_K\uparrow h$, so monotone convergence gives
\[
\int_0^1\tilde h_K(z)g(z;{\bf p})\dd z \longrightarrow \int_0^1h(z)g(z;{\bf p})\dd z.
\]
It follows that
\[
\int_0^1h_K(z)g(z;{\bf p})\dd z \longrightarrow \int_0^1h(z)g(z;{\bf p})\dd z.
\]
Therefore,
\[
\inf_{h\in\mathcal H}\int_0^1h(z)g(z;{\bf p})\dd z = \inf_{h\in\mathcal H_f}
\int_0^1h(z)g(z;{\bf p})\dd z
\]
for all ${\bf p}\in\Delta^{n-1}$.

It remains to apply Sion's Minimax Theorem. The set \(\mathcal H_f\) is convex. Indeed, the entropy constraint is preserved under convex combinations by concavity of the logarithm, admissibility is preserved under convex combinations, and finiteness of all integrals against \(g_r\) is preserved by linearity. The simplex \(\Delta^{n-1}\) is compact and convex.

Define
\[
\Psi({\bf p},h) = \int_0^1h(z)g(z;{\bf p})\dd z.
\]
For fixed $h\in\mathcal H_f$, $\Psi({\bf p},h)$ is linear and continuous in ${\bf p}$. For fixed ${\bf p}\in\Delta^{n-1}$,
\[
g(z;{\bf p}) = \sum_{r=1}^{n-1}p_rg_r(z) \le \bar g(z),
\]
and hence, for any $h,\tilde h\in L^1_{\bar g}$,
\[
\left|\Psi({\bf p},h)-\Psi({\bf p},\tilde h)\right|\le
\int_0^1|h(z)-\tilde h(z)|\bar g(z)\dd z = \|h-\tilde h\|_{\bar g}.
\]
Thus $\Psi({\bf p},h)$ is a continuous linear functional of $h$ under the weighted $L^1$ topology. All conditions of Sion's Minimax Theorem \citep[see, e.g.,][Ch. 6]{Ekeland-Temam:1999} are, therefore, satisfied, and
\[
\sup_{{\bf p}\in\Delta^{n-1}} \inf_{h\in\mathcal H_f} \Psi({\bf p},h)
= \inf_{h\in\mathcal H_f} \sup_{{\bf p}\in\Delta^{n-1}} \Psi({\bf p},h).
\]
Combining this equality with the first part of the proof establishes the result.
\end{proof}

\subsection{Robust robustness}
\label{app_robust_robustness}

In this section, we show the robustness of the effort guarantee in (\ref{implemented_value}) in the sense of \cite{Ball-Kattwinkel:2026}, using their Proposition 2. We establish this result using the weighted $L^1$ topology introduced in Appendix \ref{app_Sion}.

\begin{lemma}
\label{cor_robust_robustness}
When quantile hazards of noise generated by the ambiguity set are taken as the payoff-relevant states and endowed with a weighted $L^1$ topology, the effort guarantee induced by the robust prize schedule is robust in the sense of \cite{Ball-Kattwinkel:2026}.
\end{lemma}

\begin{proof}
Let $\mathcal S_n$ denote the set of quantile hazards generated by distributions $F\in\mathcal R_n$, endowed with the weighted $L^1$ topology defined in Appendix \ref{app_Sion}. For a fixed prize schedule ${\bf p}$, the payoff in state $h\in\mathcal S_n$ is equilibrium effort,
\[
V_{\bf p}(h) = (c')^{-1}\left(\frac{1}{n}\int_0^1h(z)g(z;{\bf p})\dd z\right).
\]
These integrals are finite for regular noise distributions. Moreover,
\[
\left|\int_0^1[h(z)-\tilde h(z)]g(z;{\bf p})\dd z\right| \le \|h-\tilde h\|_{\bar g}.
\]
Since $(c')^{-1}$ is continuous, $V_{\bf p}$ is continuous on $\mathcal S_n$.

The payoff is also bounded. In a symmetric equilibrium, each agent receives an expected prize of $1/n$ and can obtain a nonnegative payoff by choosing zero effort. Therefore, $c(V_{\bf p}(h))\le 1/n$, and hence
\[
0\le V_{\bf p}(h)\le c^{-1}\left(\frac{1}{n}\right) \qquad\text{for all }h\in\mathcal S_n.
\]

Let $\mathcal A_n(\bar H)\subseteq\mathcal S_n$ denote the hazards generated by distributions in $\mathcal F_n(\bar H)$, and define the ambiguity set of priors $\Pi = \{\delta_h:h\in\mathcal A_n(\bar H)\}$. Then
\[
\inf_{\pi\in\Pi}\int_{\mathcal S_n} V_{\bf p}(h)\dd\pi(h) = 
\inf_{F\in\mathcal F_n(\bar H)}x^*[{\bf v},F].
\]
Thus, this formulation gives the payoff guarantee
in the original robust design problem. Because $V_{\bf p}$ is bounded and continuous, its expectation is continuous under weak convergence of priors on $\mathcal S_n$. Proposition 2 of \cite{Ball-Kattwinkel:2026}, therefore, implies that its payoff guarantee is robust. This holds for every fixed prize schedule, including the robust prize schedule.
\end{proof}

\subsection{Robust prizes for small tournaments}
\label{app:small_n}

\begin{lemma}
\label{lemma_n=4}
For $n=4$, the robust tournament awards two distinct (positive) prizes, with $v_1^*>v_2^*=v_3^*>v_4^*=0$.
\end{lemma}
\begin{proof} We will show that the system of KT conditions (\ref{KT_conditions}) is satisfied for some ${\bf p}=(p_1,0,p_3)=(p,0,1-p)$ such that $\ell_1(\mathbf p)=\ell_3(\mathbf p)=1$ while $\ell_2(\mathbf p)<1$. Using $g_1(z)=12z^2(1-z)$ and $g_3(z)=4(1-z)^3$, the binding KT conditions for $\ell_1(\mathbf p)$ and $\ell_3(\mathbf p)$ imply
\[
\Phi(p)\equiv\ell_1(\mathbf p)-\ell_3(\mathbf p)
=\int_0^1\frac{3z^2-(1-z)^2}{3pz^2+(1-p)(1-z)^2}\dd z=0.
\]
Since $\Phi(p)$ is continuous and decreasing, with 
\[
\Phi'(p)=
-\int_0^1\frac{[3z^2-(1-z)^2]^2}{[3pz^2+(1-p)(1-z)^2]^2}\dd z<0,
\]
and it can be shown that $\Phi\left(\frac12\right)=1+\frac{3}{4}\log 3 -\frac{\pi\sqrt{3}}{4}>0$ and $\Phi\left(\frac{7}{10}\right) = \frac{25}{48}\left(\log 7+\frac85-\frac{3\pi}{\sqrt7}\right)<0$, we obtain $\frac{1}{2}<p^*<\frac{7}{10}$.

Next, the equations $\ell_1(\mathbf p)=\ell_3(\mathbf p)=1$ separately give
\[
\int_0^1\frac{z^2}{3pz^2+(1-p)(1-z)^2}\dd z=\frac{1}{3},
\qquad \int_0^1\frac{(1-z)^2}{3pz^2+(1-p)(1-z)^2}\dd z=1.
\]
Using $2z(1-z)=1-z^2-(1-z)^2$, we then obtain
\[
\ell_2(\mathbf p) = 3\int_0^1\frac{z(1-z)}{3pz^2+(1-p)(1-z)^2}\dd z = \frac{3}{2}\int_0^1\frac{\dd z}{3pz^2+(1-p)(1-z)^2}-2 = \frac{3\pi}{4\sqrt{3p(1-p)}}-2.
\]
Finally, since $p(1-p)>\frac{21}{100}>\frac{\pi^2}{48}$, we conclude that $\ell_2(\mathbf p)<1$.
\end{proof}

\end{document}